\documentclass[10pt]{article}
\usepackage[english]{babel}
\usepackage[T2A]{fontenc}
\usepackage[utf8]{inputenc}
\usepackage{amsmath}
\usepackage{amsthm}
\usepackage{amssymb}
\usepackage{textcomp}
\usepackage{graphicx}
\usepackage[unicode,pdfborder={0 0 1}]{hyperref}

\newcommand{\eqlb}[2]{\begin{equation} \label{#1} #2 \end{equation}}
\newcommand{\eq}[1]{\begin{equation*} #1 \end{equation*}}
\newcommand{\eqa}[1]{\begin{eqnarray*} #1 \end{eqnarray*}}
\newcommand{\eqal}[2]{\begin{eqnarray}#2\label{#1}\end{eqnarray}}
\newcommand{\mfs}[2]{\mbox{#1$#2$}}                                           
\newcommand{\eqs}[1]{$#1$}
\newcommand{\brc}[1]{\left(#1\right)}
\newcommand{\bsq}[1]{\left[#1\right]}
\newcommand{\bfi}[1]{\left\{ #1\right\}}
\newcommand{\brs}[1]{\brc{#1}}
\newcommand{\bfis}[1]{\{#1\}}
\newcommand{\bn}[1]{#1}
\newcommand{\abs}[1]{\left|#1\right|}
\newcommand{\flr}[1]{\left\lfloor #1\right\rfloor}
\newcommand{\cl}[1]{\left\lceil #1\right\rceil}

\newcommand{\qq}{\qquad}
\newcommand{\ds}{\displaystyle}
\newcommand{\pd}[2]{\dfrac{\partial#1}{\partial#2}}

\newcommand{\wt}[1]{\widetilde{#1}}
\newcommand{\matr}[2]{\begin{array}{#1}#2\end{array}}
\newcommand{\pmat}[2]{\brc{\matr{#1}{#2}}}
\newcommand{\at}[2]{\genfrac{}{}{0pt}{}{#1}{#2}}

\newcommand{\tps}[2]{#1}

\newcommand{\dd}{\widetilde{\delta}}
\newcommand{\bld}{\boldsymbol}
\newcommand{\bm}[1]{\bld{#1}}
\newcommand{\adots}{{\mbox{\rotatebox[origin=c]{90}{$\ddots$}}}}

\newtheorem{prop}{Proposition}[section]

\newcommand{\rmi}{\textrm{i}}
\newcommand{\rme}{\textrm{e}}
\newcommand{\rmd}{\textrm{d}}
\newcommand{\Tr}{\textrm{Tr}}
\newcommand{\Or}{\boldsymbol{O}}
\newcommand{\fl}{}

\numberwithin{equation}{section}
\hoffset= - 1.0in
\title{Reduction of the Elliptic \eqs{SL(N,\mathbb C)} top}

\author{G. Aminov\\
{\em e-mail: aminov@itep.ru}\\
Institute for Theoretical and Experimental Physics, Moscow, Russia\\
Moscow Institute of Physics and Technology, Moscow, Russia\\
\\
S. Arthamonov\\
{\em e-mail: artamonov@itep.ru}\\
Institute for Theoretical and Experimental Physics, Moscow, Russia\\
Moscow Institute of Physics and Technology, Moscow, Russia\\
}

\begin{document}

\maketitle

\begin{abstract}

We propose a relation between the elliptic \eqs{SL(N,\mathbb C)} top and Toda systems and obtain a new class of integrable systems in a specific limit of the elliptic \eqs{SL(N,\mathbb C)} top. The relation is based on the Inozemtsev limit (IL) and a symplectic map from the elliptic Calogero-Moser system to the elliptic \eqs{SL(N,\mathbb C)} top. In the case when \eqs{N=2} we use an explicit form of a symplectic map from the phase space of the elliptic Calogero-Moser system to the phase space of the elliptic \eqs{SL(2,\mathbb C)} top and show that the limiting tops are equivalent to the Toda chains. In the case when \eqs{N>2} we generalize the above procedure using only the limiting behavior of Lax matrices. In a specific limit we also obtain a more general class of systems and prove the integrability in the Liouville sense of a certain subclass of these systems. This class is described by a classical $r$-matrix obtained from an elliptic $r$-matrix.

\end{abstract}

\maketitle

\tableofcontents

\section{Introduction}

We study four integrable systems whose equations of motion have a Lax representation with spectral parameter \cite{Krichever1}, \cite{Krichever2}, \cite{Dubrovin}. We consider periodic and non-periodic Toda chains, the elliptic Calogero-Moser model, and the elliptic $SL(N,\mathbb C)$ top. It was established earlier by several authors that the systems are related to each other.
In \cite{Inozemtsev} Inozemtsev has proposed a procedure (IL) giving a limit relation between the Toda chains and the elliptic Calogero-Moser model. Later, the IL was generalized and used to establish connections between other integrable systems. Chernyakov and Zotov have shown in \cite{ZotovChernyakov} that the IL applied to the \eqs{SL(N,\mathbb C)} elliptic Euler-Calogero model and the elliptic Gaudin model produces new Toda-like systems endowed with additional degrees of freedom corresponding to a coadjoint orbit in \eqs{SL(N,\mathbb C)}.
Levin, Olshanetsky, and Zotov \cite{LOZ} have constructed a singular symplectic transformation from the elliptic Calogero-Moser system to the elliptic $SL(N,\mathbb C)$ top. Using this transformation Smirnov has shown in \cite{Smirnov3} that integrable tops on the algebra \eqs{\mathfrak{sl}(N,\mathbb C)} are equivalent to the $N$-particle trigonometric and rational Calogero-Moser systems. The relations between the systems can be described by the following diagrams:
\eqlb{eq:kndiag}{\matr{ccc}{
\textrm{ECM model}&\longleftrightarrow&SL(N,\mathbb C) \textrm{ top}\\
\downarrow\textrm{IL}\\
\textrm{Toda system}\\
}}
\eq{\matr{ccc}{
\textrm{Elliptic CM model}&\longleftrightarrow&SL(N,\mathbb C) \textrm{ top}\\
\downarrow\textrm{}&&\downarrow\textrm{}\\
\textrm{Trigonometric/Rational CM model}&\longleftrightarrow&\textrm{Limiting top}\\
}}

First goal of this paper is to obtain the Toda chain from the elliptic top. This complements the diagram (\ref{eq:kndiag}) in the following way:
\eq{\matr{ccc}{
\textrm{ECM model}&\longleftrightarrow&SL(N,\mathbb C) \textrm{ top}\\
\downarrow\textrm{IL}&&\downarrow\textrm{IL}\\
\textrm{Toda system}&\longleftrightarrow&\textrm{Limiting top}\\
}}

In order to obtain a new relation we will use a procedure similar to the Inozemtsev limit. The Inozemtsev limit is a combination of the trigonometric limit, infinite shifts of particle coordinates, and rescalings of the coupling constants. To obtain a limiting system equivalent to the Toda chain it is necessary to combine the Inozemtsev limit and the infinite shift of the spectral parameter. Since the spectral parameter of the elliptic \eqs{SL(N,\mathbb C)} top is given on a complex torus \eqs{T^{2}} with moduli $\tau$,
under the trigonometric limit \eqs{Im (\tau)\rightarrow+\infty} we obtain systems with spectral parameter given on an infinite complex cylinder \eqs{\mathbb C/\mathbb Z}.

In the case of the elliptic \eqs{SL(2,\mathbb C)} top it is convenient to use an explicit form (\ref{eq:strans2}) of a symplectic map from the phase space of the elliptic Calogero-Moser system to the phase space of  the top (Subsection \ref{sec:connection}). Then the Inozemtsev shifts of the elliptic Calogero-Moser system coordinates induce the rescalings of the elliptic \eqs{SL(2,\mathbb C)} top coordinates. The equivalence between the limiting systems and the Toda chains is due to the bosonizations formulas which follow from the limit of (\ref{eq:strans2}).

To derive an explicit form of the map between the phase spaces of the Calogero-Moser system and the elliptic top in the case of the elliptic \eqs{SL(N>2,\mathbb C)} top  is not as simple as when \eqs{N=2}. That is why we use the scalings of coordinates induced by the limiting behavior of Lax matrices and thus generalize \eqs{N=2} case. Also, the scalings of coordinates satisfy an important requirement, that is the limit of the Poisson algebra of the elliptic \eqs{SL(N,\mathbb C)} top  must define a Poisson structure on the phase space of the limiting system. This Poisson structure along with the values of the Casimir functions define the symplectic submanifold, for which there is the symplectic map to the phase space of the Toda chain. Equations of motion of the limiting system have Lax representation and are equivalent to the equations of motion of the Toda chain.


Second goal of this paper is to obtain in the limit a more general class of systems and prove the integrability in the Liouville sense of a certain subclass of these systems. This class appears under specific conditions on the parameters of the limit and contains Toda chains as a special case. It is possible that further study of this class will lead to establishing a connection between integrable systems mentioned above and gauge theories. Such an approach was developed earlier. For example, in \cite{GGM1}, \cite{GGM2} Toda-like systems corresponding to the multi-component magnets were studied in the context of the low-energy effective \eqs{N=2} SUSY gauge theories.

Now we will review general facts and notation about the integrable systems under consideration.

\subsection{Elliptic \tps{$SL(N,\mathbb C)$}{SL(N,C)} top}

The elliptic \eqs{SL(N,\mathbb C)} top is an example of Euler-Arnold top \cite{Arnold}. The elliptic \eqs{SL(N,\mathbb C)} top is defined on a coadjoint orbit of the group \eqs{SL(N,\mathbb C)}:
\eqlb{eq:topps}{\mathcal R^{\mathrm{rot}}=\{\bld S\in \mathfrak{sl}(N,\mathbb C),\quad \bld S=a^{-1}\bld S^{(0)}a\},}
where $a\in SL(N,\mathbb C)$ is defined up to the left multiplication on the stationary subgroup
\eqs{G_{0}} of \eqs{\bld S^{(0)}.} The phase space \eqs{\mathcal R^{\mathrm{rot}}}
is equipped with the Kirillov-Kostant symplectic form
\eq{\omega^{\mathrm{rot}}=\Tr\brc{\bld S^{(0)}\rmd aa^{-1}\wedge\rmd aa^{-1}}.}

The Hamiltonian is defined as
\eqlb{eq:Jhamtop}{H^{\mathrm{rot}}=-\dfrac12 \Tr\bld S J(\bld S).}
Here we consider a special form of a linear operator $J$ that provides the integrability of the system
\eqa{J(\bld S)=\sum_{mn}J_{mn}s_{mn}T_{mn},\qq
J_{mn}=E_2\brc{\dfrac{m+n\tau}N,\tau},
\cr m,n\in\{0,\dots, N-1\},\qq m^2+n^2\neq0,}
where $E_2(z,\tau)$ is the second Eisenstein function (see \cite{Veyl1}) defined on the complex torus \eqs{T^2:\mathbb{C}/\brc{2\omega_1\mathbb Z+2\omega_2\mathbb Z}} with
\eqs{\omega_1=\frac12,}\quad
\eqs{\tau=\omega_2/\omega_1,}
and \eqs{s_{mn}} are coordinates in the sin-algebra basis \eqs{T_{mn}} (see 
\ref{app:veylbasis}).

The equations of motion can be written in the Lax form \cite{Reynman1}:
\eqlb{eq:toplaxeq}{\dfrac{\rmd L^{\mathrm{rot}}}{\rmd t}=N\bsq{L^{\mathrm{rot}},M^{\mathrm{rot}}}.}
Factor $N$ in (\ref{eq:toplaxeq}) comes from the definition of Lax matrices in the sin-algebra basis from 
\ref{app:veylbasis}
\eqlb{eq:toplaxmatr}{\matr{ll}{
\ds L^{\mathrm{rot}}=\sum_{m,n}s_{mn}\varphi\bsq{\at mn}(z)T_{mn},&
\varphi\bsq{\at mn}(z)=\bld\rme\brc{-\dfrac{nz}N}\phi\brc{-\dfrac{m+n\tau}N,z},
\cr
\ds M^{\mathrm{rot}}=\sum_{m,n}s_{mn}f\bsq{\at mn}(z)T_{mn},&
f\bsq{\at mn}(z)=\bld\rme\brc{-\dfrac{nz}N}\partial_u\phi(u,z)|_{u=-\frac{m+n\tau}N},
}}
where \eqs{\bld\rme\brc{z}\equiv\exp\brc{2\pi\rmi z},}\quad \eqs{\rmi\equiv\sqrt{-1},} and $\phi$ is a combination of theta-functions (see 
\ref{app:ellipticfunctions}). The Lax matrix satisfies the properties of quasi-periodicity:

\eqlb{eq:topqp}{L^{\mathrm{rot}}\brc{z+1}=T_{10}L^{\mathrm{rot}}\brc{z}T_{10}^{-1},
\qq
L^{\mathrm{rot}}\brc{z+\tau}=T_{01}L^{\mathrm{rot}}\brc{z}T_{01}^{-1}.}

Consequently, \eqs{\Tr\brc{L^{\mathrm{rot}}\brc{z}}^k} are doubly periodic functions with the poles of order up to $k$, and thus they can be expanded in the basis consisting of the second Eisenstein function and its derivatives:

\eq{\fl\Tr\brc{L^{\mathrm{rot}}\brc{z}}^k=H_{k,0}+E_2\brc{z}H_{k,2}+E_2'\brc{z}H_{k,3}+\dots+E_2^{\brc{k-2}}\brc{z}H_{k,k}.}

In this way we obtain the Hamiltonian (\ref{eq:Jhamtop})
\eqlb{eq:Htop}{H^{\mathrm{rot}}=\dfrac12H_{2,0}=\dfrac12 \Tr (L^{\mathrm{rot}})^2-\dfrac12 \Tr S^2 E_2(z,\tau).}

Poisson brackets for variables $s_{mn}$ are defined by the commutator \eqs{[T_{ab},T_{cd}]} (\ref{eq:weilcomm}) of basis elements \eqs{T_{ab}} and \eqs{T_{cd}}
\eqlb{eq:topbrackets}{\{s_{ab},s_{cd}\}=2\rmi\sin\bsq{\dfrac\pi N\brs{bc-ad}}s_{a+c,b+d}.}
Then transition to the standard basis (\ref{eq:weiltr}) gives us
\eqlb{eq:Stopbrackets}{\{ S_{ij},S_{kl}\}=N(S_{kj}\delta_{il}-S_{il}\delta_{kj}).}

Linear brackets (\ref{eq:topbrackets}), (\ref{eq:Stopbrackets}) can be written in terms of the Belavin-Drinfeld classical elliptic $r$-matrix \eqs{r(z)} \cite{Belavin1}, \cite{BelavinDrinfeld}, \cite{KulishSklyanin}. Namely,

\eqlb{eq:toprbrackets}{\bfi{L_1^{\mathrm{rot}}\brc{z_1},L_2^{\mathrm{rot}}\brc{z_2}}=\bsq{r\brc{z_1-z_2},L_1^{\mathrm{rot}}\brc{z_1}+L_2^{\mathrm{rot}}\brc{z_2}},}
where
\eq{L_1\brc{z}=L\brc{z}\otimes Id,\qq L_2\brc{z}=Id\otimes L\brc{z}.}

The classical $r$-matrix is defined by

\eqlb{eq:ellipticr}{r\brc{z}=-\sum_{m,n}\varphi\bsq{\at mn}\brc{z}T_{mn}\otimes T_{-m,-n}.}

Equation (\ref{eq:toprbrackets}) implies the involutivity of the independent coefficients \eqs{H_{k,z}}. Therefore, there are \eqs{N(N+1)/2-1} independent integrals of motion. Note that \eqs{H_{k,k},\quad k\in\{2,\dots,N\}}, are the Casimirs corresponding to the coadjoint orbit (\ref{eq:topps}).

\subsection{Elliptic Calogero-Moser system}

The elliptic Calogero-Moser system (CM) was first introduced in quantum version \cite{Calogero}, \cite{Moser}. The elliptic CM system is defined on the phase space as follows
\eq{\mathcal R^{\mathrm{CM}}=\bfi{\brc{\bld{\mathrm u},\bld{\mathrm v}},\quad \sum_{i=1}^{N}u_i=0,\quad \sum_{i=1}^{N}v_i=0}}
with the canonical symplectic form
\eq{\omega^{\mathrm{CM}}=\brc{\rmd\bld{\mathrm v}\wedge \rmd\bld{\mathrm u}}.}
The corresponding Hamiltonian is defined via
\eq{H^{\mathrm{CM}}=\sum_{i=1}^{N}\dfrac{v_i^2}2+m^2\sum_{i>j}E_2\brc{u_i-u_j,z}.}
The equations of motion defined by the Hamiltonian have the Lax representation
\eq{\dfrac{\rmd L^{\mathrm{CM}}}{\rmd t}=\bsq{L^{\mathrm{CM}},M^{\mathrm{CM}}},}
where the Lax pair can be chosen in the holomorphic form in order to construct the connection between the Calogero-Moser system and the elliptic \tps{$SL(N,\mathbb C)$}{SL(N,C)} top \cite{LOZ}
\eq{L^{\mathrm{CM}}_{ij}=\delta_{ij}v_i+m\brc{1-\delta_{ij}}\phi\brc{u_i-u_j,z},}
\eq{M^{\mathrm{CM}}_{ij}=-\delta_{ij}\sum_{k\neq j}E_2(u_j-u_k)+\left.\pd{\phi(u,z)}u\right|_{u=u_i-u_j}.}

\subsection{Connection between the Calogero-Moser system and the elliptic \tps{$SL(N,\mathbb C)$}{SL(N,C)} top}

\label{sec:connection}

The connection mentioned in the headline was established in \cite{LOZ} in the form of a singular gauge transformation
\eq{L^{\mathrm{rot}}(z)=\Xi(z)L^{\mathrm{CM}}(z)\Xi^{-1}(z).}
This transformation leads to the symplectic map
\eqlb{eq:smap}{\mathcal R^{\mathrm{CM}}\rightarrow\mathcal R^{\mathrm{rot}},\qq \brc{\bld{\mathrm u},\bld{\mathrm v}}\mapsto \bld S.}

In the case when \eqs{N=2} this map has the form
\eqlb{eq:strans2}{\left\{\matr{l}{
s_{01}=-v\dfrac{\theta_{01}(0)\theta_{01}(2u)}{\vartheta'(0)\vartheta(2u)}-
m\dfrac{\theta_{01}^2(0)}{\theta_{00}(0)\theta_{10}(0)}\dfrac{\theta_{00}(2u)\theta_{10}(2u)}{\vartheta^2(2u)},\\
\\
s_{10}=v\dfrac{\theta_{10}(0)\theta_{10}(2u)}{\vartheta'(0)\vartheta(2u)}+
m\dfrac{\theta_{10}^2(0)}{\theta_{00}(0)\theta_{01}(0)}\dfrac{\theta_{00}(2u)\theta_{01}(2u)}{\vartheta^2(2u)},\\
\\
s_{11}=-\rmi v\dfrac{\theta_{00}(0)\theta_{00}(2u)}{\vartheta'(0)\vartheta(2u)}-\rmi
m\dfrac{\theta_{00}^2(0)}{\theta_{10}(0)\theta_{01}(0)}\dfrac{\theta_{10}(2u)\theta_{01}(2u)}{\vartheta^2(2u)}.
}\right.}

\subsection{Toda systems}

Periodic and nonperiodic Toda systems of $N$ interacting particles in the center of a mass frame are defined on the phase space
\eq{\mathcal R^{\mathrm{T}}=\bfi{\brc{\bld{\mathrm u},\bld{\mathrm v}},\quad \sum_{i=1}^{N}u_i=0,\quad \sum_{i=1}^{N}v_i=0},}
with the canonical symplectic form
\eq{\omega^{\mathrm{T}}=\brc{\rmd\bld{\mathrm v}\wedge \rmd\bld{\mathrm u}}.}
The Hamiltonian of the nonperiodic system is
\eq{H^{\mathrm{AT}}=\dfrac12\sum_{i=1}^{N}v_i^2 + 4\pi^2M^2\sum_{i=1}^{N-1}\bld\rme(u_{i+1}-u_{i}),}
and of the periodic system has the form
\eq{H^{\mathrm{PT}}=\dfrac12\sum_{i=1}^{N}v_i^2 + 4\pi^2M^2\sum_{i=1}^{N}\bld{ \mathrm e}(u_{i+1}-u_{i}),\qq u_{N+1}=u_1.}
The equations of motion for both nonperiodic and periodic Toda systems can be written in the Lax form (\cite{Manakov1}, \cite{Flashka1}, \cite{Flashka2})
\eq{\dfrac{\rmd}{\rmd t}L^{\mathrm{AT}}=\bsq{L^{\mathrm{AT}},M^{\mathrm{AT}}},\qq \dfrac{\rmd}{\rmd t}L^{\mathrm{PT}}=\bsq{L^{\mathrm{PT}},M^{\mathrm{PT}}}.}

One can obtain Toda systems by applying the Inozemtsev limit to the Calogero-Moser system \cite{Inozemtsev}.

\section{Elliptic $SL(2,\mathbb C)$ top via the Inozemtsev limit}

The main idea of the technique under consideration is to treat elliptic $SL(2,\mathbb C)$ top coordinates
as functions of coordinates \eqs{(\bld{\mathrm u},\bld{\mathrm v})} of the elliptic Calogero-Moser model and apply
Inozemtsev shift of \eqs{(\bld{\mathrm u},\bld{\mathrm v})}.

\subsection{Periodic Toda system from the elliptic top}

\subsubsection{Limit of the Lax matrix and Poisson algebra}

To obtain the periodic Toda system we combine  the shift of coordinates \eqs{u=U + \tau/4},
scaling of the coupling constant \eqs{M=mq^{\frac14}\quad (q\equiv\bld\rme(\tau))},
the shift of the spectral parameter \eqs{z=\wt z+\tau/2,}
and the trigonometric limit \eqs{q\rightarrow0}. Then from (\ref{eq:strans2}) we derive

\eq{s_{10}=-\dfrac{\rmi v}{\pi}+\bld\Or(q^{\frac14}),}
\eq{s_{01}=\dfrac{M\cos(2\pi U)}{q^{\frac14}}-\dfrac{v\sin(2\pi U)}{\pi}+\bld\Or(q^{\frac14}),}
\eq{s_{11}=-\dfrac{M\sin(2\pi U)}{q^{\frac14}}-\dfrac{v\cos(2\pi U)}{\pi}+\bld\Or(q^{\frac14}).}

Coordinates of the limiting top are scaled coordinates of the elliptic $SL(2,\mathbb C)$ top
\begin{subequations}
\begin{align}
\label{eq:scaled10}&\wt s_{10}=\lim_{q\rightarrow0}s_{10}=-\dfrac{\rmi v}{\pi},\\
\label{eq:scaled01}&\wt s_{01}=\lim_{q\rightarrow0}s_{01}q^{\frac14}=M\cos(2\pi U),\\
\label{eq:scaled11}&\wt s_{11}=\lim_{q\rightarrow0}s_{11}q^{\frac14}=-M\sin(2\pi U).
\end{align}
\end{subequations}

Scaled coordinates (\ref{eq:scaled10}) -- (\ref{eq:scaled11}) form an algebra which arises via contraction of \eqs{\mathfrak{sl}\brc{2,\mathbb C}} algebra
\eqlb{eq:limalg2}{\{\wt s_{10},\wt s_{11}\}=2\rmi \wt s_{01}, \quad \{\wt s_{11},\wt s_{01}\}=0,\quad \{\wt s_{01},\wt s_{10}\}=2\rmi \wt s_{11}.}

Also, formulas (\ref{eq:scaled10}) -- (\ref{eq:scaled11}) define a symplectic map from canonical coordinates \eqs{\brc{U,v}} to the coordinates of the limiting top. Such formulas are known as bosonization formulas.

The following condition defines the symplectic leaf:

\eq{\wt s_{01}^2+\wt s_{11}^2=const=M^2,}
and originates from the Casimir function of the elliptic $SL(2,\mathbb C)$ top
\eq{s_{01}^2+s_{10}^2+s_{11}^2=const=m^2.}

Taking into account the behavior of the function \eqs{\varphi\bsq{\at mn}(z)} (\ref{eq:varphiexp}) (see 
\ref{app:ellipticfunctions})
we can write down the limiting Lax matrix
\eq{\wt L^{\mathrm{rot}}=
4\pi\brc{\matr{cc}{
\dfrac\rmi4 \wt s_{10}&\wt s_{01}\sin(\pi \wt z)-\wt s_{11}\cos(\pi \wt z)\\
\\
\wt s_{01}\sin(\pi \wt z)+\wt s_{11}\cos(\pi \wt z)&-\dfrac\rmi4 \wt s_{10}
}},}
where \eqs{\wt L^{rot}=\lim_{q\rightarrow0} L^{rot}}.

\subsubsection{Limiting equations of motion and bosonization}

We will use formula (\ref{eq:Jhamtop}) for computing the limit of Hamiltonian

\eq{\wt H^{\mathrm{rot}}=-\brc{\wt s_{01}^2\wt J_{01}+\wt s_{10}^2\wt J_{10}+\wt s_{11}^2\wt J_{11}},}
where

\eq{\wt J_{10}=\lim_{q\rightarrow0}J_{10}=\pi^2,}
\eq{\wt J_{01}=\lim_{q\rightarrow0}J_{01}q^{-\frac12}=-8\pi^2,}
\eq{\wt J_{11}=\lim_{q\rightarrow0}J_{11}q^{-\frac12}=8\pi^2.}

The series expansion of \eqs{f\bsq{\at mn}(z)} (\ref{eq:fexp}) (see 
\ref{app:ellipticfunctions}) leads to the second Lax matrix

\eq{\wt M^{\mathrm{rot}}=\pi^2\brc{\matr{cc}{
\wt s_{10}&4\brc{\wt s_{01}+\rmi\wt s_{11}}\bld\rme\brc{-\dfrac {\wt z}2}\\
4\brc{\wt s_{01}-\rmi\wt s_{11}}\bld\rme\brc{-\dfrac {\wt z}2}&-\wt s_{10}\\
}}.}

Equations of motion (\ref{eq:toplaxeq}) preserve the same form in the limit:

\eq{\dfrac{\rmd\wt L^{\mathrm{rot}}}{\rmd t}=\bfi{\wt H^{\mathrm{rot}},\wt L^{\mathrm{rot}}}=2\bsq{\wt L^{\mathrm{rot}},\wt M^{\mathrm{rot}}}.}

Using bosonization formulas (\ref{eq:scaled10}), (\ref{eq:scaled01}), and (\ref{eq:scaled11}) one can obtain the periodic Toda system

\eq{\wt H^{\mathrm{rot}}\rightarrow H^{\mathrm{PT}}=v^2+8M^2\pi^2\cos\brc{4\pi U},}

\eq{\wt L^{\mathrm{rot}}\rightarrow L^{\mathrm{PT}}=\brc{\matr{cc}{
v&4\pi M\sin\brc{\pi\brc{2U+\wt z}}\\
\\
-4\pi M\sin\brc{\pi\brc{2U-\wt z}}&-v\\
}},}

\eq{\wt M^{\mathrm{rot}}\rightarrow M^{\mathrm{PT}}=\brc{\matr{cc}{
-\rmi\pi v&4\pi^2 M\bld\rme\brc{-U-\dfrac{\wt z}2}\\
4\pi^2 M\bld\rme\brc{U-\dfrac{\wt z}2}&\rmi\pi v\\
}}.
}

\subsection{Nonperiodic Toda system from the elliptic top}

\subsubsection{Limit of the Lax matrix and Poisson algebra}

Here we will use another shift of coordinates \eqs{u=U + \tau/8},
another scaling of the coupling constant \eqs{M=mq^{1/8}},
the same shift of the spectral parameter \eqs{z=\wt z+\tau/2}
and apply the trigonometric limit \eqs{q\rightarrow0}.
That gives us the following formulas for the elliptic \eqs{SL(2,\mathbb C)} top coordinates

\eq{s_{10}=-\dfrac{\rmi v}{\pi}+\bld\Or(q^{\frac18}),}

\eq{s_{01}=\dfrac{M \bld\rme(U)}{2q^{\frac14}}+\dfrac{\rmi v\bld\rme(U)}{2\pi q^{\frac18}}+\bld\Or(1),}

\eq{s_{11}=\dfrac{\rmi M \bld\rme(U)}{2q^{\frac14}}-\dfrac{v\bld\rme(U)}{2\pi q^{\frac18}}+\bld\Or(1),}
and for the coordinates of the limiting top
\eq{\wt s_{10}=\lim_{q\rightarrow0}s_{10}=-\dfrac{\rmi v}{\pi},}
\eq{\wt s_{01}=\lim_{q\rightarrow0}s_{01}q^{\frac14}=\dfrac12 M \bld\rme(U),}
\eq{\wt s_{11}=\lim_{q\rightarrow0}s_{11}q^{\frac14}=\dfrac\rmi2 M \bld\rme(U).}
These coordinates form the same algebra (\ref{eq:limalg2}) as in the case of the periodic Toda system.

The limiting Lax matrix has the form
\eq{\wt L^{\mathrm{rot}}=
4\pi\brc{\matr{cc}{
\dfrac\rmi4 \wt s_{10}&\wt s_{01}\sin(\pi \wt z)-\wt s_{11}\cos(\pi \wt z)\\
\\
\wt s_{01}\sin(\pi \wt z)+\wt s_{11}\cos(\pi \wt z)&-\dfrac\rmi4 \wt s_{10}
}}.}
In this case the symplectic leaf of the elliptic \eqs{SL(2,\mathbb C)} top turns into
\eqlb{eq:apsymleaf}{\wt s_{01}^2+\wt s_{11}^2=0.}

\subsubsection{Limiting equations of motion and bosonization}

Since we have the same shift of the spectral parameter and the same scaling of coordinates \eqs{s_{mn}} as in the periodic case, \eqs{J_{mn}} have equivalent limits

\eq{\wt J_{10}=\pi^2,\quad
\wt J_{01}=-\wt J_{11}=-8\pi^2.}
Also, limiting Hamiltonian and the second Lax matrix acquire the same forms

\eq{\wt H^{\mathrm{rot}}=-\brc{\wt s_{01}^2\wt J_{01}+\wt s_{10}^2\wt J_{10}+\wt s_{11}^2\wt J_{11}},}
\eq{\wt M^{\mathrm{rot}}=\pi^2\brc{\matr{cc}{
\wt s_{10}&4\brc{\wt s_{01}+\rmi\wt s_{11}}\bld\rme\brc{-\dfrac {\wt z}2}\\
4\brc{\wt s_{01}-\rmi\wt s_{11}}\bld\rme\brc{-\dfrac {\wt z}2}&-\wt s_{10}\\
}}.}
The limiting Hamiltonian can be simplified on the symplectic leaf (\ref{eq:apsymleaf}) as follows

{\eq{\wt H^{\mathrm{rot}}=-\brc{2\wt s_{01}^2\wt J_{01}+\wt s_{10}^2\wt J_{10}}.}
The equations of motion have the Lax representation

\eq{\dfrac{\rmd\wt L^{\mathrm{rot}}}{\rmd t}=\bfi{\wt H^{\mathrm{rot}},\wt L^{\mathrm{rot}}}=2\bsq{\wt L^{\mathrm{rot}},\wt M^{\mathrm{rot}}}.}
Bosonization formulas transform the limiting top into the nonperiodic Toda system

\eq{\wt H^{\mathrm{rot}}\rightarrow H^{\mathrm{AT}}=v^2+4M^2\pi^2\bld\rme(2U),}
\eq{\wt L^{\mathrm{rot}}\rightarrow L^{\mathrm{AT}}=\brc{\matr{cc}{
v&-2\rmi \pi M\bld\rme(U+\frac{\wt z}2)\\
2\rmi \pi M\bld\rme(U-\frac{\wt z}2)&-v
}},}
\eq{\wt M^{\mathrm{rot}}\rightarrow M^{\mathrm{AT}}=\brc{\matr{cc}{
-\rmi \pi v&0\\
4\pi^2 M\bld\rme(U-\frac{\wt z}2)&\rmi\pi v
}}.}

\section{Elliptic \tps{$SL(N>2,\mathbb C)$}{SL(N>2,C)} top via the Inozemtsev limit}

In this section we consider a limit that is a combination of the shift of the spectral parameter \mbox{$z=\wt z+\tau/2$}, the scalings of coordinates, and the trigonometric limit \eqs{Im(\tau)\rightarrow +\infty}. The scalings of coordinates are defined by the limiting behavior of the Lax matrix and are not derived from the symplectic map (\ref{eq:smap}) as in the case $N=2$. Also, the scalings of coordinates satisfy an important requirement, that is the limit of the Poisson algebra of the elliptic \eqs{SL(N,\mathbb C)} top must define a Poisson structure on the phase space of the limiting system. This Poisson structure along with the values of Casimir functions define the symplectic submanifold for which there is the symplectic map to the phase space of the Toda chain. Equations of motion of the limiting system have Lax representation and are equivalent to the equations of motion of the Toda chain.


\subsection{Periodic Toda system from the elliptic top}
\label{sub:todaper}
\subsubsection{Limit of the Lax matrix and Poisson algebra}

In order to determine the exact scaling of coordinates we need to expand the function \eqs{\varphi\bsq{\at mn}(z)} as series in $q$, where \eqs{q=e^{2\pi\rmi\tau}} (see 
\ref{app:ellipticfunctions}). We obtain

\eqlb{varphiseries1}{\varphi\bsq{\at mn}\brc{\wt z+\dfrac\tau2}=
\left\{\matr{ll}{
-\pi\bld\rme\brc{\dfrac m{2N}}\sin^{-1}\brc{\pi\dfrac mN}+\bld{\mathrm o}(1),&n=0,\\
\\
2\pi\rmi\bld\rme\brc{-\dfrac{n\wt z}N+\dfrac mN}q^{\frac n{2N}}+\bld{\mathrm o}\brc{q^{\frac n{2N}}},&0<n<\dfrac N2,\\
\\
\bld\Or\brc{q^{\frac 14}},&n=\dfrac N2,\\
\\
-2\pi\rmi\bld\rme\brc{\dfrac{N-n}N\wt z}q^{\frac 12-\frac n{2N}}+\bld{\mathrm o}\brc{q^{\frac 12-\frac n{2N}}},&\dfrac N2<n<N.
}\right.
}

Since the Lax matrix for the periodic Toda system can be written in tridiagonal form, the following substitution is reasonable
\eqal{eq:scaling1}{&&s_{mn}=\wt s_{mn} q^{-g(n)},\qq m,n\in\{0,\dots, N-1\},\quad m^2+n^2\neq0,
\cr
&&g(n)=\frac{1-\wt\delta\brc{n}}{2N}.}
This gives us the limiting matrix $\wt L^{\mathrm{rot}}$
\eqal{eq:limitingperL}{&&\wt L^{\mathrm{rot}}=-\pi\sum_{m=1}^{N-1}\bld\rme\brc{\dfrac m{2N}}\sin^{-1}\brc{\pi\dfrac mN} \wt s_{m0}T_{m0}+
\cr
&&+2\pi\rmi\sum_{m=0}^{N-1}\bsq{\bld\rme\brc{-\dfrac{\wt z}N+\dfrac mN}\wt s_{m1}T_{m1}-\bld\rme\brc{\dfrac{\wt z}N}\wt s_{m,N-1}T_{m,N-1}},}
where

\eq{\wt \delta(n)=\left\{\matr{ll}{
1&n\equiv0\bmod N,\\
0&n\not\equiv0\bmod N.\\
}\right.}

As one can see, coordinates \eqs{\wt s_{mn},}\quad \eqs{1<n<N-1,} are not present in the adduced matrix and hence in the Hamiltonian. We will show later that Hamilton equations for these variables can be integrated in spite of the fact that their dynamics are separated from the Lax representation.

After scaling (\ref{eq:scaling1}), we obtain the contraction of Poisson algebra (\ref{eq:topbrackets}) in the limit \eqs{q\rightarrow 0}
\eqlb{eq:scperbrackets}{\{\wt s_{ab},\wt s_{cd}\}=2\rmi\sin\bsq{\dfrac\pi N\brs{bc-ad}}q^{g(b)+g(d)-g(b+d)}\wt s_{a+c,b+d}.}
Therefore, scaled coordinates $\wt s_{mn}$ with the Poisson brackets form an algebra in the limit of $q\rightarrow0$ provided that
\eqlb{eq:brcond}{\forall k,n:\quad g(k)+g(n)-g(k+n)\geqslant0.}

If \eqs{g(n)=\brc{1-\wt\delta\brc{n}}/\brc{2N}}, then (\ref{eq:brcond}) is trivial and we can write down all nonzero brackets corresponding to the equality in (\ref{eq:brcond})
\eq{\{\wt s_{a0},\wt s_{cd}\}=-2\rmi\sin\brc{\dfrac\pi N ad}\wt s_{a+c,d}.}

It is convenient to use the standard basis further. In this basis substitution (\ref{eq:scaling1}) and the Lax matrix turn into
\eq{S_{ij}=\wt S_{ij} q^{-\mathfrak g(i,j)},\qq \mathfrak g(i,j)=\frac{1-\delta_{ij}}{2N}, \qq i,j\in \{1,\dots, N\},}
\eqal{eq:SlimitingperL}{&&\wt L^{\mathrm{rot}}_{ij}=\dfrac{2\pi\rmi}N\sum_{m=1}^{N}\sum_{k=1}^{N-1}
\wt S_{mm}\bld\rme\brc{\dfrac{k(i-m)}N}\brc{\bld\rme\brc{-\dfrac kN}-1}^{-1}\delta_{ij}+\cr
&&+2\pi\rmi\wt S_{i+1,i+2}\bld\rme\brc{-\dfrac{\wt z}N}\wt \delta\brc{j-i-1}-
2\pi\rmi\wt S_{i,i-1}\bld\rme\brc{\dfrac{\wt z}N}\wt \delta\brc{j-i+1}.}

From now on, indexes of $\wt S$ belong to \eqs{\bfi{1,\dots,N}} and satisfy the properties of periodicity
\eq{\wt S_{N+i,j}=\wt S_{i,N+j}=\wt S_{ij}\qq i,j\in \mathbb Z.}

From (\ref{eq:Stopbrackets}) we obtain the following nonzero Poisson brackets for coordinates in the standard basis:
\eqlb{eq:LStopbr}{\{\wt S_{ii},\wt S_{jk}\}=N(\wt S_{ji}\delta_{ik}-\wt S_{ik}\delta_{ij}).}
Now, it can be easily seen that Casimir functions are
\eqal{eq:Casimirgeneral}{
&&\sum_{i=1}^{N}\wt S_{ii},
\cr
&&\wt S_{i_1,i_2}\wt S_{i_2,i_3}\dots\wt S_{i_k,i_1}\qq\forall j\neq l\quad i_l\neq i_j,\qq 2\leqslant k\leqslant N.}

In (\ref{eq:Casimirgeneral}) there are $N+2$ independent functions depending only on variables that form the Lax matrix (\ref{eq:SlimitingperL})
\eq{\sum_{i=1}^{N}\wt S_{ii},\qq
\prod_{i=1}^{N}\wt S_{i,i-1},\qq
\wt S_{i,i+1}\wt S_{i+1,i}\quad i\in\{1,\dots, N\},}

Also, in (\ref{eq:Casimirgeneral}) there are \eqs{N\brc{N-3}} Casimir functions independent as a functions of variables which are not included in the Lax matrix

\eq{\brc{\prod_{j=1}^{k}\wt S_{i+j-1,i+j}}\wt S_{i+k,i}\qq 1\leqslant i\leqslant N,\qq 2\leqslant k\leqslant N-2.}

Thus, on the symplectic submanifold with nonzero values of quadratic Casimir functions \eqs{\wt S_{i,i+1}\wt S_{i+1,i}} variables which are not included in the Lax matrix are the following functions of variables included in the Lax matrix

\eq{\wt S_{i+k,i}=const\prod_{j=1}^{k}\wt S_{i+j,i+j-1}\qq 1\leqslant i\leqslant N,\qq 2\leqslant k\leqslant N-2.}

\subsubsection{Limiting equations of motion and bosonization}

The limit of Hamiltonian, as it follows from (\ref{eq:Htop}) and the fact that
\eqs{\Tr S^2 E_2\brc{z,\tau}\rightarrow 0}, depends only on variables contributed in the Lax matrix (\ref{eq:limitingperL})
\eqal{eq:periodicham}{&&\wt H^{\mathrm{rot}}=\dfrac12\Tr(\wt L^{\mathrm{rot}})^2=
-\dfrac{\pi^2}N\sum_{m,n=1}^{N}\sum_{k=1}^{N-1}\wt S_{mm}\wt S_{nn}\bld\rme\brc{\dfrac{k(n-m)}N}\brc{1-\cos\brc{2\pi\dfrac kN}}^{-1}+
\cr
&&+4\pi^2\sum_{i=1}^{N}\wt S_{i,i+1}\wt S_{i,i-1}.}

Using the series expansion of \eqs{f\bsq{\at mn}(\wt z+\tau/2)} (see 
\ref{app:ellipticfunctions})

\eq{\fl f\bsq{\at mn}\brs{\wt z+\dfrac\tau2}=
\left\{\matr{ll}{
-\pi^2\sin^{-2}\brc{\pi\dfrac mN}+\bld{\mathrm o}(1),&n=0,\\
\\
4\pi^2\bld\rme\brc{\dfrac mN}\bld\rme\brc{-\dfrac{n\wt z}N}q^{\frac n{2N}}+\bld{\mathrm o}\brc{q^{\frac n{2N}}},&0<n<\dfrac{3N}4,\\
\\
4\pi^2\bsq{\bld\rme\brc{\dfrac mN}-\bld\rme\brc{-\dfrac nN+\wt z}}\bld\rme\brc{-\dfrac34\wt z}q^{\frac38}+\bld{\mathrm o}\brc{q^{\frac38}},&n=\dfrac{3N}4,\\
\\
-4\pi^2\bld\rme\brc{-\dfrac mN+\wt z}\bld\rme\brc{-\dfrac{n\wt z}N}q^{\frac32\brc{1-\frac nN}}+\bld{\mathrm o}\brc{q^{\frac32\brc{1-\frac nN}}},&\dfrac{3N}4<n<N,\\
}\right.}
one can obtain the second Lax matrix

\eqlb{eq:SlimitingperM}{\fl\wt M^{\mathrm{rot}}_{ij}=
-\frac{\pi^2}{N}\delta_{ij}\sum_{m=1}^{N-1}\sum_{k=1}^{N}\sin^{-2}\brc{\pi\dfrac mN}\bld\rme\brc{\dfrac{m(i-k)}N}\wt S_{kk}+4\pi^2\wt\delta\brc{j-i-1}\wt S_{i+1,i+2}\bld\rme\brc{-\dfrac{\wt z}N},}
and ensure that the equations of motion can be written in the Lax form

\eqlb{eq:LaxRep}{\dfrac{\rmd}{\rmd t}\wt L^{\mathrm{rot}}=\{\wt H^{\mathrm{rot}},\wt L^{\mathrm{rot}}\}=N\bsq{\wt L^{\mathrm{rot}},\wt M^{\mathrm{rot}}}.}

Those variables that are not included in the Lax matrix (\ref{eq:SlimitingperL}) have simple dynamics

\eqal{eq:simdyn}{&&\dfrac{\rmd}{\rmd t}\wt S_{ij}=4\pi^2\wt S_{ij}\sum_{m=1}^{N}\sum_{k=1}^{N-1}\wt S_{mm}\sin\brc{\pi\dfrac{k(j-i)}{N}}\sin\brc{\pi\dfrac{k(i+j-2m)}{N}}\times
\cr
&&\times\brc{1-\cos\brc{2\pi\dfrac kN}}^{-1},\qq
1<\brc{j-i}\bmod N<N-1,}
and with other coordinates allow bosonization formulas

\eqal{eq:bosonization}{
&&\wt S_{ii}=\dfrac N{2\pi \rmi}(v_{i-1}-v_i),
\cr
&&\wt S_{i,i+1}=MN\bld\rme(u_i),
\cr
&&\wt S_{i+1,i}=MN\bld\rme(-u_i),
\cr
&&\wt S_{i,i+k}=c_{i,i+k}\bld e\brc{\sum_{n=i}^{i+k-1}u_n}\quad 2\leqslant k\leqslant N-2,\quad c_{i,i+k}=const,}
where \eqs{\bld{\mathrm u},\bld{\mathrm v}} are canonical coordinates

\eq{\{v_i,u_j\}=\delta_{ij}\quad i,j\in\{1,\dots, N\},}
and
\eq{\sum_{i=1}^N u_i=0,\qq\sum_{i=1}^N v_i=0.}

Let us show that for \eqs{\bld{\mathrm u},\bld{\mathrm v}} we have dynamics of the periodic Toda system in the center of a mass frame. Substituting (\ref{eq:bosonization}) into Hamiltonian (\ref{eq:periodicham}) and Lax matrices (\ref{eq:SlimitingperL}), (\ref{eq:SlimitingperM}) we obtain

\eq{\wt H^{\mathrm{rot}}=N^2\sum_{i=1}^{N}\dfrac{v_i^2}2+4\pi^2 M^2 N^2\sum_{i=1}^{N}\bld\rme\brc{u_{i+1}-u_i}=N^2 H^{\mathrm{PT}},}
where $H^{\mathrm{PT}}$ has the form of periodic Toda Hamiltonian,

\eq{\fl\wt L^{\mathrm{rot}}=2\pi\rmi MN
\arraycolsep=0pt
\mfs{}{\left(\matr{cccccc}{
\dfrac{v_1}{2\pi\rmi M}&\bld\rme(u_2-\dfrac{\wt z}N)&0&\dots&0&-\bld\rme(-u_N+\dfrac{\wt z}N)\\
-\bld\rme(-u_1+\dfrac{\wt z}N)&\dfrac{v_2}{2\pi\rmi M}&\bld\rme(u_3-\dfrac{\wt z}N)&\dots&0&0\\
0&-\bld\rme(-u_2+\dfrac{\wt z}N)&\dfrac{v_3}{2\pi\rmi M}&\dots&0&0\\
\dots&\dots&\dots&\dots&\dots&\dots\\
0&0&0&\dots&\dfrac{v_{N-1}}{2\pi\rmi M}&\bld\rme(u_N-\dfrac{\wt z}N)\\
\bld\rme(u_1-\dfrac{\wt z}N)&0&0&\dots&-\bld\rme(-u_{N-1}+\dfrac{\wt z}N)&\dfrac{v_n}{2\pi\rmi M}
}\right)},}

\eqal{eq:periodicm}{&&\wt M^{\mathrm{rot}}_{ij}=
\dfrac{\rmi\pi}2
\sum_{m=1}^{N-1}\sum_{k=1}^{N}\sin^{-2}\brc{\pi\dfrac mN}\bld\rme\brc{\dfrac{m(i-k)}N}\brc{v_{k-1}-v_k}\delta_{ij}+
\cr
&&+4\pi^2MN\bld\rme\brc{u_{i+1}-\dfrac{\wt z}N}\wt\delta\brc{j-i-1}.}

After the gauge transformation we have

\eqal{eq:gauge}{&&\wt L^{\mathrm{rot}}\rightarrow g^{-1}\wt L^{\mathrm{rot}} g,\qq
\wt M^{\mathrm{rot}}\rightarrow g^{-1}\wt M^{\mathrm{rot}} g+\frac{1}{N}g^{-1}\dot g,
\cr
&&g_{ij}=\delta_{ij}\bld\rme\brc{\dfrac{i\wt z}N}\prod_{k=1}^{i-1}\bld\rme\brc{-u_k},}
and Lax matrices take the standard form

\eq{\fl\wt L^{\mathrm{rot}}=2\pi\rmi MN
\left(\matr{cccccc}{
\dfrac{v_1}{2\pi\rmi M}&\bld\rme(u_2-u_1)&0&\dots&0&-\bld\rme\brs{\wt z}\\
-1&\dfrac{v_2}{2\pi\rmi M}&\bld\rme(u_3-u_2)&\dots&0&0\\
0&-1&\dfrac{v_3}{2\pi\rmi M}&\dots&0&0\\
\dots&\dots&\dots&\dots&\dots&\dots\\
0&0&0&\dots&\dfrac{v_{N-1}}{2\pi\rmi M}&\bld\rme(u_N-u_{N-1})\\
\bld\rme(u_1-u_N-\wt z)&0&0&\dots&-1&\dfrac{v_n}{2\pi\rmi M}
}\right),}
\vspace{\baselineskip}
\eq{\fl\wt M^{\mathrm{rot}}=4\pi^2MN
\left(\matr{cccccc}{
0&\bld\rme(u_2-u_1)&0&\dots&0&0\\
0&0&\bld\rme(u_3-u_2)&\dots&0&0\\
0&0&0&\dots&0&0\\
\dots&\dots&\dots&\dots&\dots&\dots\\
0&0&0&\dots&0&\bld\rme(u_N-u_{N-1})\\
\bld\rme(u_1-u_N-\wt z)&0&0&\dots&0&0
}\right).}

\subsection{Nonperiodic Toda system from the elliptic top}

\subsubsection{Limit of Lax matrices and Poisson algebra}

To obtain the Lax matrix of the nonperiodic Toda chain we are going to consider another substitution

\eqal{eq:apersub}{S_{ij}=\wt S_{ij} q^{-\mathfrak g(i,j)},
\qq \mathfrak g(i,j)=\dfrac{1-\delta_{ij}-\dfrac12\delta_{i1}\delta_{jN}}{2N}
\qq i,j\in \{1,\dots, N\}.}
Hence, the contraction of Poisson algebra (\ref{eq:Stopbrackets}) in the limit \eqs{q\rightarrow0} takes the form

\eq{\{\wt S_{ij},\wt S_{kl}\}=Nq^{\mathfrak g(i,j)+\mathfrak g(k,i)}\brc{
\delta_{il}\wt S_{kj}q^{-\mathfrak g(k,j)}-\delta_{kj}\wt S_{il}q^{-\mathfrak g(i,l)}},}
and the scaled coordinates form an algebra in the limit \eqs{q\rightarrow0} if

\eq{\forall i,j,k\qq \mathfrak g(i,j)+\mathfrak g(k,i)-\mathfrak g(k,j)\geqslant0.}
As it can be easily seen the above inequality is valid for \eqs{\mathfrak g(i,j)} defined in (\ref{eq:apersub}). Upon the limit \eqs{q\rightarrow 0} we have nonzero brackets (\ref{eq:LStopbr}) and the following Lax matrices:

\eqal{eq:SlimitingaperL}{&&\wt L^{\mathrm{rot}}_{ij}=\dfrac{2\pi\rmi}N\sum_{m=1}^{N}\sum_{k=1}^{N-1}
\wt S_{mm}\bld\rme\brc{\dfrac{k(i-m)}N}\brc{\bld\rme\brc{-\dfrac kN}-1}^{-1}\delta_{ij}+\cr
&&+2\pi\rmi\wt S_{i+1,i+2}\bld\rme\brc{-\dfrac{\wt z}N}\wt \delta\brc{j-i-1}-
2\pi\rmi\wt S_{i,i-1}\bld\rme\brc{\dfrac{\wt z}N} \delta_{i,j+1},}

\eqal{eq:SlimitingaperM}{&&\wt M^{\mathrm{rot}}_{ij}=
-\frac{\pi^2}{N}\delta_{ij}\sum_{m=1}^{N-1}\sum_{k=1}^{N}\sin^{-2}\brc{\pi\dfrac mN}\bld\rme\brc{\dfrac{m(i-k)}N}\wt S_{kk}+
\cr
&&+4\pi^2\wt\delta\brc{j-i-1}\wt S_{i+1,i+2}\bld\rme\brc{-\dfrac{\wt z}N}.}

As long as the algebra has the same limit as in the periodic case we have Casimir functions (\ref{eq:Casimirgeneral}). But now there are \eqs{N+1} independent functions formed only by variables contributed in the Lax matrix

\eq{\sum_{i=1}^{N}\wt S_{ii},\qq
\prod_{i=1}^{N}\wt S_{i,i+1},\qq
\wt S_{i,i+1}\wt S_{i+1,i}\quad i\in\{1,\dots, N-1\}}
and \eqs{N(N-3)+1} Casimir functions independent as a functions of variables which are not included in the Lax matrix

\eqa{
&&\wt S_{1,N}\wt S_{N,1},
\cr
&&\brc{\prod_{j=1}^{k}\wt S_{i+j-1,i+j}}\wt S_{i+k,i}\qq 1\leqslant i\leqslant N,\qq 2\leqslant k\leqslant N-2.}

\subsubsection{Limiting equations of motion and bosonization}

The limit of Hamiltonian after substitution (\ref{eq:apersub}) is

\eqal{eq:nonperiodicham}{&&\wt H^{\mathrm{rot}}=\dfrac12\Tr(\wt L^{\mathrm{rot}})^2=-\dfrac{\pi^2}N\sum_{m,n=1}^{N}\sum_{k=1}^{N-1}\wt S_{mm}\wt S_{nn}\bld\rme\brc{\dfrac{k(n-m)}N}\brc{1-\cos\brc{2\pi\dfrac kN}}^{-1}+\cr
&&+4\pi^2\sum_{i=2}^{N}\wt S_{i,i+1}\wt S_{i,i-1}.}

The equations of motion can be written in the Lax form

\eq{\dfrac{\rmd}{\rmd t}\wt L^{\mathrm{rot}}=\{\wt H^{\mathrm{rot}},\wt L^{\mathrm{rot}}\}=N\bsq{\wt L^{\mathrm{rot}},\wt M^{\mathrm{rot}}}.}
These equations imply simple dynamics (\ref{eq:simdyn}) for variables that are not included in the Lax matrix. And for all coordinates of the limiting system there are bosonization formulas

\eqal{eq:bosonisation2}{
&&\wt S_{ii}=\dfrac N{2\pi \rmi}(v_{i-1}-v_i),\quad i\in\{1,\dots, N\},
\cr
&&\wt S_{i,i+1}=MN\bld\rme(u_i),\quad i\in\{1,\dots, N\},
\cr
&&\wt S_{i+1,i}=MN\bld\rme(-u_i),\quad i\in\{1,\dots, N-1\},
\cr
&&\wt S_{1,N}=const\,\bld e\brc{-u_N},
\cr
&&\wt S_{i,i+k}=c_{i,i+k}\bld e\brc{\sum_{n=i}^{i+k-1}u_n},\quad2\leqslant k\leqslant N-2,\quad c_{i,i+k}=const.
}

Canonical coordinates \eqs{\bld{\mathrm u}, \bld{\mathrm v}} have dynamics of the nonperiodic Toda chain in the center of a mass frame. After substituting (\ref{eq:bosonisation2}) into Hamiltonian (\ref{eq:nonperiodicham}) we obtain

\eq{\wt H^{\mathrm{rot}}=N^2\sum_{i=1}^{N}\dfrac{v_i^2}2+4\pi^2 M^2 N^2\sum_{i=1}^{N-1}\bld\rme\brc{u_{i+1}-u_i}=N^2 H^{\mathrm{AT}},}
where $H^{\mathrm{AT}}$ has the form of nonperiodic Toda Hamiltonian. Lax matrices take the usual form under gauge transformation (\ref{eq:gauge}) mentioned in the periodic case

\eq{\fl\wt L^{\mathrm{rot}}=2\pi\rmi MN
\left(\matr{cccccc}{
\dfrac{v_1}{2\pi\rmi M}&\bld\rme(u_2-u_1)&0&\dots&0&0\\
-1&\dfrac{v_2}{2\pi\rmi M}&\bld\rme(u_3-u_2)&\dots&0&0\\
0&-1&\dfrac{v_3}{2\pi\rmi M}&\dots&0&0\\
\dots&\dots&\dots&\dots&\dots&\dots\\
0&0&0&\dots&\dfrac{v_{N-1}}{2\pi\rmi M}&\bld\rme(u_N-u_{N-1})\\
\bld\rme(u_1-u_N-\wt z)&0&0&\dots&-1&\dfrac{v_n}{2\pi\rmi M}
}\right),}
\vspace{\baselineskip}
\eq{\fl\wt M^{\mathrm{rot}}=4\pi^2MN
\left(\matr{cccccc}{
0&\bld\rme(u_2-u_1)&0&\dots&0&0\\
0&0&\bld\rme(u_3-u_2)&\dots&0&0\\
0&0&0&\dots&0&0\\
\dots&\dots&\dots&\dots&\dots&\dots\\
0&0&0&\dots&0&\bld\rme(u_N-u_{N-1})\\
\bld\rme(u_1-u_N-\wt z)&0&0&\dots&0&0
}\right).}

\subsection{More general class of limiting systems}
\label{sec:somesytem}

In the previous Subsections we have considered substitutions of variables (\ref{eq:scaling1}) and (\ref{eq:apersub}), which after applying the Inozemtsev limit lead to the Toda chains. It turns out that the substitutions mentioned above are not the only possibility to provide the integrable systems in the limit. We will consider the following generalization:

\eqal{eq:perscaling}{&s_{mn}=\wt s_{mn} q^{-g(n)},\qq m,n\in\{0,\dots, N-1\},
\quad m^2+n^2\neq0,
\cr
&\qq g(i)=\left\{\matr{ll}{
\dfrac k{2N},& 0\leqslant k\leqslant p<\dfrac N2,\\
\\
\dfrac p{2N},&p < k < N-p,\\
\\
\dfrac {N-k}{2N},&N-p\leqslant k<N,
}\right.}
where \eqs{i\in \mathbb Z}, \eqs{k\equiv i\bmod N}, and prove the integrability of the limiting systems in the case when $N$ and $p$ are relatively prime.

\subsubsection{Limit of Lax matrices and Poisson algebra}
\label{sec:lma}

Scaled coordinates with Poisson brackets (\ref{eq:scperbrackets}) form a Poisson algebra in the limit \eqs{q\rightarrow0} provided that the following condition is valid

\eqlb{eq:pineq}{\forall k,n:\quad g(k)+g(n)-g(k+n)\geqslant0.}
For \eqs{g(n)} under consideration (\ref{eq:pineq}) is proved in 
\ref{app:inequality}.

Nonzero limiting brackets have the form

\eqal{eq:lpbrackets}{
&&\{\wt s_{a0},\wt s_{cd}\}=-2\rmi\sin\brc{\dfrac\pi N ad}\wt s_{a+c,d},
\cr\cr
&&\fl\{\wt s_{ab},\wt s_{cd}\}=2\rmi\sin\brc{\dfrac\pi N\brc{bc-ad}}\wt s_{a+c,b+d}
\cr
&&\quad(0<a\leqslant p)\wedge(0<b\leqslant p)\wedge(0<a+b\leqslant p),
\cr\cr
&&\fl\{\wt s_{ab},\wt s_{cd}\}=2\rmi\sin\brc{\dfrac\pi N\brc{bc-ad}}\wt s_{a+c,b+d}
\cr
&&\quad(N-p\leqslant a<N)\wedge(N-p\leqslant b<N)\wedge(2N-p\leqslant a+b<2N),
}
or in the standard basis
\eqal{eq:lpstbrackets}{
&&\fl\{\wt S_{ii},\wt S_{jk}\}=N(\wt S_{ji}\delta_{ik}-\wt S_{ik}\delta_{ij}),
\cr\cr
&&\fl\{\wt S_{ij},\wt S_{kl}\}=N(\wt S_{kj}\delta_{il}-\wt S_{il}\delta_{kj})
\cr
&&(0<(j-i)\bmod N\leqslant p)\wedge(0<(l-k)\bmod N\leqslant p)\wedge(0<(j+l-i-k)\bmod N\leqslant p),
\cr\cr
&&\fl\{\wt S_{ij},\wt S_{kl}\}=N(\wt S_{kj}\delta_{il}-\wt S_{il}\delta_{kj})
\cr
&&(N-p\leqslant(j-i)\bmod N< N)\wedge(N-p\leqslant(l-k)\bmod N<N)\wedge
\cr
&&\wedge(N-p\leqslant(j+l-i-k)\bmod N<N),
}
where as usual \eqs{\wedge} stands for ``and''.

Formulas (\ref{eq:lpbrackets}) (or (\ref{eq:lpstbrackets})) imply that the limiting Poisson algebra is solvable. Thus, there is no general method to construct all Casimir functions, but in the special case when $N$ and $p$ are relatively prime we able to present the whole set of independent Casimir functions.

At first we are interested in Casimir functions in general case when \eqs{p<N/2}. Since elements \eqs{\wt S_{i,i+k},\quad k\in\{p,\dots,N-p\}}, have nonzero brackets only with coordinates \eqs{\wt S_{ii},\;\wt S_{i+k,i+k}}, we obtain the second and $N$'th order Casimir functions

\eqal{eq:psimplecas}{
&&\wt S_{i,i+k}\wt S_{i+k,i}\quad k\in\bfi{p,\dots,\flr{\frac N2}},\quad i\in\{1,\dots,N\},
\cr
&&\prod_{i=1}^{N}\wt S_{i,i+k}\quad k\in\{p,\dots,N-p\},\quad i\in\{1,\dots,N\},}
where \eqs{\flr{x}} is the floor function of $x$.

In the case when $N$ and $p$ are relatively prime we can also construct the following \eqs{N(N-2p-1)} independent Casimir functions

\eq{\fl\brc{\prod_{j=1}^k\wt S_{i+(j-1)p,i+jp}}\wt S_{i+kp,i}\quad p<kp \bmod N<N-p,\quad k,i\in\{1,\dots,N\}.}

If \eqs{p<\brc{N-1}/2} it is convenient to consider two disjoint subalgebras. The first one is generated by the variables

\eq{\mathcal S_1=\bfi{\wt S_{ij},\;p<\brc{j-i}\bmod N<N-p}}
with simple dynamics, which we are going to show further. Equations of motion for the elements of the second subalgebra

\eq{\fl\mathcal S_2=\bfi{\wt S_{ij},\;(0\leqslant\brc{j-i}\bmod N\leqslant p)\;\textrm{or}\; (N-p\leqslant \brc{j-i}\bmod N<N)}}
have Lax representation, so we need to obtain the number of independent Casimir functions from universal enveloping algebra of this subalgebra. Since Casimir functions lower the dimension of the symplectic submanifold, we need to majorize the degeneration factor of Poisson tensor \eqs{\pi^{(ij)(kl)}\brs{\mathcal S_2}}:

\eq{\bfi{F\brs{\mathcal S_2},G\brs{\mathcal S_2}}=\pi^{(ij)(kl)}\brs{\mathcal S_2}\;\partial_{(ij)}F\;\partial_{(kl)}G,\qq \partial_{(ij)}=\pd{}{\wt S_{ij}}.}

Formula (\ref{eq:lpstbrackets}) implies that Poisson tensor can be represented as a block
lower matrix with respect to antidiagonal (\ref{eq:pitens}). Rank $R$ of this
\eqs{(2p+1)N\times(2p+1)N}-matrix satisfies the condition \eqs{R\geqslant 2p(N-1)\;} (\ref{eq:symprank}) (see 
\ref{app:dsl}), which restricts the number of the independent Casimir functions up to \eqs{N+2p} (here we treat \eqs{\ds\sum_{i=1}^{N}\wt S_{ii}} as a Casimir).

Linear brackets (\ref{eq:lpbrackets}) and (\ref{eq:lpstbrackets}) can be written in terms of an $r$-matrix. Namely,

\eqlb{eq:prbrackets}{\bfi{\wt L_1^{\mathrm{rot}}\brc{\wt z_1},\wt L_2^{\mathrm{rot}}\brc{\wt z_2}}=\bsq{r\brc{\wt z_1-\wt z_2},\wt L_1^{\mathrm{rot}}\brc{\wt z_1}+\wt L_2^{\mathrm{rot}}\brc{\wt z_2}},}
where

\eq{\wt L_1\brc{z}=\wt L\brc{z}\otimes Id,\qq \wt L_2\brc{z}=Id\otimes \wt L\brc{z},}
and matrix \eqs{\wt r\brc{\wt z}} is the limit of elliptic $r$-matrix (\ref{eq:ellipticr})

\eqa{&&\wt r\brc{\wt z_1-\wt z_2}=\lim_{Im\brc{\tau}\rightarrow+\infty}r\brc{z_1-z_2}
=\pi\sum_{m=1}^{N-1}\brc{\cot\dfrac {\pi m}N-\cot\brc{\pi(\wt z_1-\wt z_2)}}T_{m0}\otimes T_{-m,0}-
\cr
&&-\pi\sin^{-1}\brc{\pi(\wt z_1-\wt z_2)}\bld\rme\brc{\dfrac{\wt z_1-\wt z_2}2}
\sum_{n=1}^{N-1}\bld\rme\brc{-\dfrac{n(\wt z_1-\wt z_2)}N}\sum_{m=0}^{N-1}T_{mn}\otimes T_{-m,-n}.}
Explicit expression for the elements of \eqs{\wt r\brc{\wt z}}

\eqa{&&\wt r_{(ii_1),(jj_1)}(\wt z)=\pi\sum_{m=1}^{N-1}\bld\rme\brc{\dfrac{m(i-i_1)}N}\brc{\cot\dfrac{\pi m}N+\rmi}\delta_{ij}\delta_{i_1j_1}+
\cr
&&+\dfrac{2\pi\rmi N \bld\rme(\wt z)}{1-\bld\rme(\wt z)}\bld\rme\brc{-\wt z\dfrac{(i_1-i)\bmod N}N}\delta_{ij_1}\delta_{i_1j},}
where we exclude one summand proportional to \eqs{Id\otimes Id}.

Substitution (\ref{eq:perscaling}) preserves in the limit coordinates of the Lax matrix with respect to the following sin-algebra basis elements:

\eq{T_{mn},\quad n\in\{-p,\dots,p\},}
which gives

\eqal{eq:lmatrixp}{&&\wt L^{\mathrm{rot}}=-\pi\sum_{m=1}^{N-1}\bld\rme\brc{\dfrac m{2N}}\sin^{-1}\brc{\pi\dfrac mN}\wt s_{m0}T_{m0}+
\cr
&&+2\pi\rmi\sum_{l=1}^p\sum_{m=0}^{N-1}\bsq{\bld\rme\brc{-\dfrac{l\wt z}N+\dfrac mN}\wt s_{ml}T_{ml}-\bld\rme\brc{\dfrac{l\wt z}N}\wt s_{m,-l}T_{m,-l}}.}
In the standard basis the Lax matrix acquires the form

\eqa{&&\wt L_{ij}^{\mathrm{rot}}=2\pi\rmi\sum_{l=1}^p\brc{\wt S_{i+1,i+1+l}\bld\rme\brc{-\dfrac{l\wt z}N}\dd\brc{j-i-l}-\wt S_{i+l,i}\bld\rme\brc{\dfrac{l\wt z}N}\dd\brc{j-i+l}}+
\cr
&&+\dfrac{2\pi\rmi}N\sum_{m=1}^{N}\sum_{k=1}^{N-1}\wt S_{mm}\bld\rme\brc{\dfrac{k\brc{i-m}}N}\brc{\bld\rme\brc{-\dfrac kN}-1}^{-1}\delta_{ij}.}

It is convenient to use both the above and gauge transformed form of the Lax matrix

\eqlb{eq:gaugetrp}{\wt L^{\mathrm{rot}}_g=g\wt L^{\mathrm{rot}}g^{-1},\qq g_{ij}=\delta_{ij}\bld\rme\brc{\dfrac{i\wt z}N}.}
Denoting \eqs{w=\bld\rme\brc{\wt z}} we obtain

\eq{\wt L^{\mathrm{rot}}_g=2\pi\rmi\mfs{\small}{
\arraycolsep=0pt
\brc{
\matr{cccccccccccccc}{
l_1&\wt S_{23}&\dots&\wt S_{2,2+p}&0&\dots&0&-w\wt S_{1,N-p+1}&\dots&-w\wt S_{1N}\\
\\
-\wt S_{21}&l_2&\wt S_{34}&\dots&\wt S_{3,3+p}&0&\dots&0&\ddots&\vdots\\
\\
\vdots&-\wt S_{32}&\ddots&\ddots&\dots&\ddots&0&\dots&0&-w\wt S_{pN}\\
\\
-\wt S_{p+1,1}&\dots&\ddots&\ddots&\wt S_{p+2,p+3}&\dots&\ddots&0&\dots&0\\
\\
0&-\wt S_{p+2,2}&\dots&-\wt S_{p+2,p+1}&l_{p+2}&\ddots&\dots&\ddots&0&\vdots\\
\\
\vdots&0&\ddots&\dots&\ddots&\ddots&\ddots&\dots&\ddots&0\\
\\
0&\dots&0&\ddots&\dots&\ddots&l_{N-p}&\wt S_{N-p+1,N-p+2}&\dots&\wt S_{N-p+1,1}\\
\\
\dfrac{\wt S_{N-p+2,2}}w&0&\dots&0&\ddots&\dots&-\wt S_{N-p+1,N-p}&\ddots&\ddots&\vdots\\
\\
\vdots&\ddots&0&\dots&0&\ddots&\dots&\ddots&l_{N-1}&\wt S_{N,1}\\
\\
\dfrac{\wt S_{N+1,2}}w&\dots&\dfrac{\wt S_{1,p+1}}w&0&\dots&0&-\wt S_{N,N-p}&\dots&-\wt S_{N,N-1}&l_N\\

}}}
,
}

where

\eq{l_i=\frac1{2\pi\rmi}L_{ii}=\dfrac1N\sum_{m=1}^{N}\sum_{k=1}^{N-1}\wt S_{mm}\bld\rme\brc{\dfrac{k\brc{i-m}}N}\brc{\bld\rme\brc{-\dfrac kN}-1}^{-1},}
and consequently

\eqlb{eq:lstrans}{\wt S_{ii}=l_{i-1}-l_i.}

\subsubsection{Lax representation of limiting equations of motion}

Limiting Hamiltonian and the second Lax matrix have similar structures as in the case when the limiting system is equivalent to the periodic Toda chain (see Subsection \ref{sub:todaper}), namely

\eqal{eq:psecham}{&&\wt H^{\mathrm{rot}}=\dfrac12\Tr\brc{\wt L^{\mathrm{rot}}}^2=-\dfrac{\pi^2}N\sum_{m,n=1}^{N}\sum_{k=1}^{N-1}\wt S_{mm}\wt S_{nn}\bld\rme\brc{\dfrac{k\brc{n-m}}N}\brc{1-\cos\brc{2\pi\dfrac kN}}^{-1}+\cr
&&+4\pi^2\sum_{l=1}^p\sum_{i=1}^N\wt S_{i+l,i}\wt S_{i+1,i+l+1},}

\eqa{&&\wt M^{\mathrm{rot}}_{ij}=-\dfrac{\pi^2}N\ds\sum_{m=1}^N\ds\sum_{k=1}^{N-1}\wt S_{mm}\sin^{-2}\brc{\pi\dfrac kN}\bld\rme\brc{\dfrac{k(i-m)}N}\delta_{ij}+
\cr
&&+4\pi^2\ds\sum_{l=1}^p\wt S_{i+1,i+l+1}\bld\rme\brc{-\dfrac{l\wt z}N}\wt\delta(j-i-l).}
When \eqs{p=1} these formulae turn into (\ref{eq:periodicham}) and (\ref{eq:periodicm}), respectively.

After the limit equations of motion also have Lax representation

\eqlb{eq:PLaxRep}{\dfrac{\rmd}{\rmd t}\wt L^{\mathrm{rot}}=\{\wt H^{\mathrm{rot}},\wt L^{\mathrm{rot}}\}=N\bsq{\wt L^{\mathrm{rot}},\wt M^{\mathrm{rot}}}.}

In the case when \eqs{p<\brc{N-1}/2} there are variables \eqs{\mathcal S_1} which are not included in the Lax pair. Hamilton equations for these variables are

\eqa{&&\dfrac{\rmd}{\rmd t}\wt S_{ij}=4\pi^2\wt S_{ij}\sum_{m=1}^{N}\sum_{k=1}^{N-1}\wt S_{mm}\sin\brc{\pi\dfrac{k(j-i)}{N}}\sin\brc{\pi\dfrac{k(i+j-2m)}{N}}\times
\cr
&&\times\brc{1-\cos\brc{2\pi\dfrac kN}}^{-1},\qq p<(j-i)\bmod N<N-p
.}

On the symplectic submanifold with nonzero values of the second order Casimir functions
\eqs{\wt S_{i,i+p}\wt S_{i+p,i},\quad i\in\{1,\dots,N\}} in the case when $N$ and $p$ are relatively prime variables $\mathcal S_1$ are the following functions of variables $\mathcal S_2$

\eq{\fl\wt S_{i+kp,i}=const \prod_{i=1}^k\wt S_{i+jp,i+(j-1)p},\quad p<kp\bmod N<N-p,\quad k,i\in\{1,\dots,N\}.}
Thus, if we solve the equations of motion (\ref{eq:PLaxRep}), we immediately obtain the solutions of the equations of motion for variables $\mathcal S_1$.

\subsubsection{Integrability}

The Lax operator of the elliptic \eqs{SL\brc{N,\mathbb C}} top satisfies properties of quasi-periodicity (\ref{eq:topqp}). Namely,

\eq{L^{\mathrm{rot}}\brc{z+1}=T_{10}L^{\mathrm{rot}}\brc{z}T_{10}^{-1},\qq
L^{\mathrm{rot}}\brc{z+\tau}=T_{01}L^{\mathrm{rot}}\brc{z}T_{01}^{-1}.}

After taking the trigonometric limit \eqs{Im\brc{\tau}\rightarrow+\infty} the Lax operator has only one quasi-period

\eq{\wt L^{\mathrm{rot}}\brc{\wt z+1}=T_{10}\wt L^{\mathrm{rot}}\brc{\wt z}T_{10}^{-1}.}

Since \eqs{\Tr\brc{\wt L^{\mathrm{rot}}\brc{\wt z}}^k} are periodic functions in \eqs{\wt z}, they can be expanded in Fourier basis \eqs{\bfi{\bld\rme\brc{j\wt z}\equiv w^j,\;j\in\mathbb Z}}. From the gauge transformed Lax matrix \eqs{\wt L_g^{\mathrm{rot}}} it follows that there are finite number of nonzero coefficients in this expansion.

\begin{prop}
\label{prop:trform}
The trace of the $k$-th power of the Lax matrix has the form

\eqlb{eq:prop1}{\fl\Tr\brc{\wt L^{\mathrm{rot}}\brc{\wt z}}^k=\sum_{j=-M}^{M}H_{kj}w^j,\quad\textrm{where}\quad M=\flr{\dfrac{kp}N},\; w\equiv\bld\rme\brc{\wt z}.}

\end{prop}

\begin{proof}
Replacing \eqs{\bld\rme\brc{\wt z}} by $w$ in formula (\ref{eq:lmatrixp}) we obtain a convenient form of the Lax matrix

\eqa{
&&\wt L^{\mathrm{rot}}=-\pi\sum_{m=1}^{N-1}\bld\rme\brc{\dfrac m{2N}}\sin^{-1}\brc{\pi\dfrac mN}\wt s_{m0}T_{m0}
+2\pi\rmi\sum_{n=1}^p\sum_{m=0}^{N-1}\bsq{\bld\rme\brc{\dfrac mN} w^{-\frac nN}\wt s_{mn}T_{mn}-w^{\frac nN}\wt s_{m,-n}T_{m,-n}}=
\cr
&&=\sum_{n=-p}^p\sum_{m=0}^{N-1}\bld c\brc{m,n}\wt s_{mn}w^{-\frac nN}T_{mn}.}
Then

\eqlb{eq:genlk}{\fl\brc{\wt L^{\mathrm{rot}}\brc{\wt z}}^k=\sum_{m_1,n_1}\ldots\sum_{m_k,n_k}w^{-\sum_{i=1}^k\frac{n_i}N}\prod_{i=1}^k\bld c\brc{m_i,n_i}\wt s_{m_in_i}T_{m_in_i}.}

By the properties of \eqs{T_{mn}} (see 
\ref{app:veylbasis}) the following condition holds

\eqlb{eq:prop1cond}{\Tr\brc{\prod_{i=1}^{k}T_{m_in_i}}\neq0\quad\Rightarrow\quad\sum_{i=1}^k n_i\equiv0\bmod N\quad\Leftrightarrow
\quad\sum_{i=1}^k\dfrac{n_i}N\;\in\mathbb Z.}

As \eqs{n_i\in\bfi{-p,\dots,p}} for any $i$, we derive the second condition \eqs{\ds\abs{\sum_{i=1}^{k}n_i}\leqslant kp}, which along with (\ref{eq:prop1cond}) implies (\ref{eq:prop1}).
\end{proof}

Representation (\ref{eq:prbrackets}) of linear brackets (\ref{eq:lpbrackets}) and (\ref{eq:lpstbrackets}) provides us with the involutivity of the coefficients \eqs{H_{kj}}.

\begin{prop}$\cite{Kheshin}.$
\label{prop:involution}
The coefficients \eqs{H_{kj}} are in involution, i.e.,
\eqlb{eq:prop2}{\bfi{H_{k_1j_1},H_{k_2j_2}}=0.}
\end{prop}

\begin{proof}

Exactly as in the case of the elliptic top we have

\eq{\fl\bfi{\Tr\brc{\wt L^{\mathrm{rot}}(\wt z_1)}^{k_1},\Tr\brc{\wt L^{\mathrm{rot}}(\wt z_2)}^{k_2}}=
\Tr\bfi{\brc{\wt L^{\mathrm{rot}}(\wt z_1)}^{k_1},\brc{\wt L^{\mathrm{rot}}(\wt z_2)}^{k_2}}.}

Then, it follows from (\ref{eq:prbrackets}) that these functions Poisson commute. Using the expansion (\ref{eq:prop1}) we get the involutivity of the coefficients (\ref{eq:prop2}).

\end{proof}

In particular, all functions \eqs{H_{kj}} Poisson commute with the Hamiltonian (\ref{eq:psecham}). Moreover, among coefficients \eqs{H_{kj}} we have Casimir functions of the form \eqs{H_{k(j),\pm j},\;k(j)=\cl{j\frac{N}p},\;j\in\{1,\dots,p\}} (\eqs{\cl{x}} is the ceiling function of $x$). The latter statement is proved in Proposition \ref{prop:casimirs} below.  It is possible to visualize these Casimirs and integrals of motion in the form of the following triangle

\eqlb{eq:pyramid}{\fl\mfs{\small}{\matr{ccccccccc}{
&&&&H_{20}&&&&
\\
&&&&\vdots&&&&
\\
&&&\bld H_{\cl{\frac{N}p},-1}&H_{\cl{\frac{N}p},0}&\bld H_{\cl{\frac{N}p},1}&&&
\\
&&&\vdots&\vdots&\vdots&&&
\\
&&\bld H_{\cl{\frac{2N}p},-2}&H_{\cl{\frac{2N}p},-1}&H_{\cl{\frac{2N}p},0}&H_{\cl{\frac{2N}p},1}&\bld H_{\cl{\frac{2N}p},2}&&
\\
&&\vdots&\vdots&\vdots&\vdots&\vdots&&
\\
&\vdots&\vdots&\vdots&\vdots&\vdots&\vdots&\vdots&
\\
\bld H_{N,-p}&\vdots&\vdots&\vdots&\vdots&\vdots&\vdots&\vdots&\bld H_{Np}
}}}
If $N$ and $p$ are relatively prime, Casimir functions \eqs{H_{N,\pm p}} are proportional to those introduced in (\ref{eq:psimplecas})

\eq{H_{N,\pm p}\propto\prod_{i=1}^{N}\wt S_{i,i\mp p}.}

Also, one can note that \eqs{H_{N,p}},\;\eqs{H_{N,-p}} and the second order Casimir functions \eqs{\wt S_{i,i+p}\wt S_{i+p,i},\; i\in\bfi{1,\dots,N}}, are not independent.

\begin{prop}
\label{prop:casimirs}
The coefficients \eqs{H_{k(j),\pm j},\;k(j)=\cl{jN/p},\;j\in\{1,\dots,p\},} of expansion (\ref{eq:prop1}) are Casimir functions.
\end{prop}

\begin{proof}

Functions \eqs{H_{k(j),\pm j}} are the coefficients of the terms \eqs{w_1^{\pm j}\equiv \bld\rme\brc{\pm j \wt z_1}} of the expansion of \eqs{\Tr \brc{\wt L^{\mathrm{rot}}(\wt z_1)}^{k(j)}} in (\ref{eq:prop1}). To prove the statement of the proposition we will show that terms \eqs{w_1^{\pm j}} are not present in the expansion of \eqs{\bfi{\Tr \brc{\wt L^{\mathrm{rot}}(\wt z_1)}^{k(j)},\wt L^{\mathrm{rot}}(\wt z_2)}} as a series in \eqs{w_1}. Using representation (\ref{eq:prbrackets}) of linear brackets (\ref{eq:lpbrackets}), (\ref{eq:lpstbrackets}) and the following form of $r$-matrix

\eq{\wt r\brc{\wt z}=\sum_{m,n}\wt \varphi\bsq{\at mn}\brc{\wt z}T_{mn}\otimes T_{-m,-n},}
where

\eq{\fl\wt\varphi\bsq{\at mn}(\wt z)=\left\{\matr{lll}{
\pi\cot\brc{\dfrac{\pi m}N}-\cot\brc{\pi\wt z},&n=0,&0<m<N,\\
\\
-\dfrac{\pi}{\sin\brc{\pi\wt z}}\bld\rme\brc{\dfrac{\wt z}2-\dfrac{n\wt z}N},&0<n<N,&0\leqslant m<N,
}\right.}
we obtain

\eqa{&&\bfi{\brc{\wt L^{\mathrm{rot}}(\wt z_1)}^k,\wt L^{\mathrm{rot}}(\wt z_2)}=\sum_{m,n}\wt\varphi\bsq{\at mn}(\wt z_1-\wt z_2)\bsq{T_{mn},\brc{\wt L^{\mathrm{rot}}(\wt z_1)}^k}\otimes T_{-m,-n}+
\cr
&&+\sum_{mn}\wt\varphi\bsq{\at mn}(\wt z_1-\wt z_2)\brc{\sum_{i=1}^k \brc{\wt L^{\mathrm{rot}}(\wt z_1)}^{i-1}T_{mn}\brc{\wt L^{\mathrm{rot}}(\wt z_1)}^{k-i}}
\otimes\bsq{T_{-m,-n},\wt L^{\mathrm{rot}}(\wt z_2)}.}
This gives

\eqal{eq:trbrackets}{&&\bfi{\Tr \brc{\wt L^{\mathrm{rot}}(\wt z_1)}^k,\wt L^{\mathrm{rot}}(\wt z_2)}
=k\sum_{m,n}\wt \varphi\bsq{\at mn}(\wt z_1-\wt z_2)\Tr\brc{\brc{\wt L^{\mathrm{rot}}(\wt z_1)}^{k-1}T_{mn}}\times
\cr
&&\times\bsq{T_{-m,-n},\wt L^{\mathrm{rot}}(\wt z_2)}.}

Substituting the explicit expressions of Sin-Algebra basis elements \eqs{T_{mn}} and \eqs{\wt\varphi\bsq{\at mn}(\wt z)} into formula (\ref{eq:trbrackets}) we get

\eqal{eq:kbrexp}{&&\bfi{\Tr \brc{\wt L^{\mathrm{rot}}(\wt z_1)}^k,\wt L^{\mathrm{rot}}(\wt z_2)_{i_1 j_1}}
=\pi k\sum_{m=1}^{N-1}\sum_{i=1}^N \cot\brc{\dfrac{\pi m}N} \bld\rme\brc{\dfrac{i m}N} \brc{\wt L^{\mathrm{rot}}(\wt z_2)}_{i_1 j_1}\times
\cr
&&\times\brc{\brc{\wt L^{\mathrm{rot}}(\wt z_1)}^{k-1}}_{ii}\brc{\bld\rme\brc{-\dfrac{i_1 m}N}-\bld\rme\brc{-\dfrac{j_1 m}N}}+
\cr
&&+\pi\rmi kN\brc{\wt L^{\mathrm{rot}}(\wt z_2)}_{i_1 j_1}\sum_{i=1}^{N}\brc{\brc{\wt L^{\mathrm{rot}}(\wt z_1)}^{k-1}}_{i i}\brc{\delta_{ii_1}-\delta_{ij_1}}-
2\pi\rmi kN\dfrac{w_1}{w_1-w_2}\bsq{K,\wt L^{\mathrm{rot}}(\wt z_2)}_{i_1 j_1},}
where

\eqa{&& K_{i_1j_1}=\bld\rme\brc{-\dfrac{\brc{(i_1-j_1)\bmod N}\brc{\wt z_1-\wt z_2}}N}\brc{\brc{\wt L^{\mathrm{rot}}(\wt z_1)}^{k-1}}_{i_1j_1}=
\cr
&&=\brc{\dfrac{w_2}{w_1}}^{\frac{(i_1-j_1)\bmod N}N}\brc{\brc{\wt L^{\mathrm{rot}}(\wt z_1)}^{k-1}}_{i_1j_1}.}

By the properties of \eqs{T_{mn}}

\eq{\brc{\prod_l T_{m_l n_l}}_{i_1j_1}\propto\bld\rme\brc{\dfrac{i_1\sum_l m_l}N}\wt \delta\brc{i_1+\sum_l n_l-j_1}.}
Therefore, each element of matrix \eqs{\bld K} is a Laurent polynomial in \eqs{w_1}, which can be seen from expression (\ref{eq:genlk}). More precisely,

\eqa{&& K_{i_1j_1}=\sum_{m_1,n_1}\ldots\sum_{m_{k-1},n_{k-1}}w_2^{\frac{(i_1-j_1)\bmod N}N} w_1^{-\frac{(i_1-j_1)\bmod N}N-\sum_{l=1}^{k-1}\frac{n_l}N}\times
\cr
&&\times\wt \delta\brc{i_1+\sum_{l=1}^{k-1}n_l-j_1}\prod_{l=1}^{k-1}\bld c'\brc{m_l,n_l}\wt s_{m_ln_l},}
and the degree of $w_1$ is the integer

\eqal{eq:kdegw}{&&-\dfrac{(i_1-j_1)\bmod N}N-\sum_{l=1}^{k-1}\dfrac{n_l}N=-\dfrac{(i_1-j_1)\bmod N}N
-\dfrac{x N+(j_1-i_1)\bmod N}N=
\cr
&&=-(x+1-\delta_{i_1j_1}),\qq x\in\mathbb Z.}

To derive the maximum \eqs{x_{\mathrm{max}}} and minimum \eqs{x_{\mathrm{min}}} of $x$ we use the condition

\eq{\abs{\sum_{l=1}^{k-1}n_l}=\abs{xN+(j_1-i_1)\bmod N}\leqslant (k-1)p.}
Consequently,

\eq{x_{\mathrm{max}}=\flr{\dfrac{(k-1)p-(j_1-i_1)\bmod N}N},
\qq
x_{\mathrm{min}}=-\flr{\dfrac{(k-1)p+(j_1-i_1)\bmod N}N}.}

In the case when \eqs{k=\cl{jN/p}} we have \eqs{\flr{(k-1)p/N}=j-1} and hence

\eq{x_{\mathrm{max}}\leqslant j-1,\qq x_{\mathrm{min}}\geqslant -j.}
Moreover, \eqs{x_{\mathrm{min}}=1-j} for the diagonal elements of matrix \eqs{\bld K}. According to formula (\ref{eq:kdegw}) we obtain the maximum and minimum degrees $d_{\mathrm{max}}$ and $d_{\mathrm{min}}$ of $w_1$  in matrix \eqs{\bld K}, respectively,

\eq{d_{\mathrm{min}}=-j,\qq d_{\mathrm{max}}=j-1.}
Hence, each element of matrix \eqs{\bsq{K,\wt L^{\mathrm{rot}}(\wt z_2)}} is a Laurent polynomial in $w_1$ with maximum and minimum degrees $d_{\mathrm{max}}$ and $d_{\mathrm{min}}$. Due to the fact that

\eq{\bld K|_{w_1=w_2}=\brc{\wt L^{\mathrm{rot}}(\wt z_2)}^{k-1},}
polynomials \eqs{w_1^j\bsq{K,\wt L^{\mathrm{rot}}(\wt z_2)}_{i_1j_1}} have root \eqs{w_1=w_2} and

\eq{-2\pi\rmi N\dfrac{w_1}{w_1-w_2}\bsq{K,\wt L^{\mathrm{rot}}(\wt z_2)}_{i_1 j_1}}
(the last term in (\ref{eq:kbrexp})) is a Laurent polynomial in $w_1$ with degrees from the interval \eqs{\bsq{1-j,j-1}}.
Also, the other two terms in (\ref{eq:kbrexp}) contain only the diagonal elements of matrix \eqs{\brc{\wt L^{\mathrm{rot}}(\wt z_1)}^{k-1}} and thus they are Laurent polynomials in $w_1$ with degrees from the same interval \eqs{\bsq{1-j,j-1}}. Therefore, terms \eqs{w_1^{\pm j}} are not present in the expansion (\ref{eq:kbrexp}).
\end{proof}

Furthermore, it is convenient to treat \eqs{H_{kj}} as functions of the following set of variables

\eq{\fl\mathcal S=\bfi{l_i,\quad i\in\bfi{1,\dots,N},\quad\sum_i l_i=0;\qq S_{jk},\quad j\neq k,\quad j,k\in\bfi{1,\dots,N}},}
where there are \eqs{N-1} independent coordinates $l_i$ and \eqs{N-1} independent coordinates $S_{jj}$, wich are connected through nondegenerate transformation (\ref{eq:lstrans}).

\begin{prop}
\label{prop:independence}
If $N$ and $p$ are relatively prime, then the coefficients \eqs{H_{kj},\;k>0,\;\flr{j N/p}\leqslant k\leqslant N,\;\abs{j}<p,} the second order Casimir functions \eqs{\wt S_{i,i+p}\wt S_{i+p,i},} \quad \eqs{1\leqslant i\leqslant N,} and \eqs{H_{Np}} are functionally independent.
\end{prop}

\begin{proof}

Coefficients \eqs{H_{kj}} are the $k$'th order homogeneous polynomials in variables \eqs{\wt S_{mn}}. To prove the independence we are going to consider the terms with the maximum degree of the variables \eqs{\bfi{l_i,\;1\leqslant i\leqslant N}}. These terms in turn contain ones with the maximum degree of the variables \eqs{\bfi{\wt S_{i,i\pm p},\;1\leqslant i\leqslant N}}. Let us denote these terms by \eqs{H_{kj}'}. Then the independence of \eqs{H_{kj}} follows from the independence of \eqs{H_{kj}'}. Moreover, the statement of the proposition follows from the independence of \eqs{H_{kj}',\;k>0,\;\flr{j N/p}\leqslant k\leqslant N,\;\abs{j}<p,} and the second order Casimir functions (the second order Casimir functions are monomials and remain the same after taking terms with the maximum degree of any set of variables). After taking the leading terms in Casimirs we obtain

\eq{H_{10}'\propto\sum_{i=1}^{N}l_i,}

\eqlb{eq:casterms}{H_{k(j),\mp j}'\propto\sum_{i=1}^{N}\brc{\prod_{m=1}^{k(j)-1}\wt S_{i\pm(m-1)p,i\pm mp}}\wt S_{i\pm(k(j)-1)p,i},
\qq k(j)=\cl{\dfrac{Nj}p},\qq 0<j<p.}

By the properties of Casimirs as the coefficients in (\ref{eq:prop1}), \eqs{H_{kj}'} are homogeneous polynomials in $l_i$'s and the summands of (\ref{eq:casterms}), so it is convenient to treat these summands as new variables

\eq{ X_{i\pm(k(j)-1)p,i}=\brc{\prod_{m=1}^{k(j)-1}\wt S_{i\pm(m-1)p,i\pm mp}}\wt S_{i\pm(k(j)-1)p,i},\qq
1\leqslant i\leqslant N,\qq 0<j<p.}
The expression

\eq{\wt S_{i\pm(k(j)-1)p,i}=\wt S_{i\pm \flr{\frac{Nj}p}p,i}=\wt S_{i\pm Nj\mp (Nj)\bmod p,i}=\wt S_{i\mp(Nj)\bmod p,i}}
implies the independence of $X$ in the case when $N$ and $p$ are relatively prime and \eqs{{0<j<p}}.

Therefore, we introduce the set of variables

\eq{\fl\mathcal S^1=\bfi{l_i,\;1\leqslant i\leqslant N;\quad X_{i,i\pm j},\;(1\leqslant i\leqslant N)\wedge(0<j<p);\quad \wt S_{i,i\pm p},\;1\leqslant i\leqslant N}}
and expand the differentials of the functions \eqs{H_{kj}'}, \quad \eqs{1<k\leqslant N,} \quad \eqs{\abs{j}<p,}\quad \eqs{H_{Np}'}, and the second order Casimirs in the basis of the differentials of \eqs{\mathcal S^1}. We treat the differentials of \eqs{l_i,} \quad \eqs{1\leqslant i\leqslant N}, as independent due to the presence of \eqs{H_{10}'}, which determines the connection between them. That leads to the Jacobian matrix of the map from \eqs{\mathcal S^1} to Hamiltonians and Casimirs. The independence of functions under consideration follows directly from the fact that the Jacobian matrix has maximum rank. After appropriate ordering of variables, Hamiltonians, and Casimirs one can write the Jacobian matrix in a block lower triangular form

\eqlb{eq:Jacobi}{\fl
\bordermatrix{
&l_1\dots l_N&\dots&\bn{\bld X(j)}&\bn{\bld X(-j)}&\dots&\wt S_{pN}&\bfi{\wt S_{i,i+p}}\cr
\bn{\bld H'(0)}
&\bld V&\bld0&\bld0&\bld0&\bld0&\bld0&\bld0\cr
\vdots&\vdots&\ddots&\bld0&\bld0&\bld0&\bld0&\bld0\cr
\bn{\bld H'(j)}
&\vdots&\dots&\bld A\brc{j}&\bld0&\bld0&\bld0&\bld0\cr
\bn{\bld H'(-j)}
&\vdots&\dots&\dots&\bld A(j)&\bld0&\bld0&\bld0\cr
\cr
\vdots&\vdots&\dots&\dots&\dots&\ddots&\bld0&\bld0\cr
\cr
H_{Np}'&\bld0&\bld0&\bld0&\bld0&\bld0&b&\bld0\cr
\cr
\bn{\bld C_2}&\bld0&\bld0&\bld0&\bld0&\bld0&\vdots&\bld D\cr
}.}

Since we care about the rank of the Jacobian matrix, without loss of generality we can neglect numerical common factors in each \eqs{H_{k,j}'} and use the ordering

\eq{
\bld H'(0)=\brc{\matr{c}{
H_{10}'\\
\vdots\\
H_{N0}'\\
}};\qq
\bn{\bld H'(j)}=\brc{\matr{c}{
H_{k(\abs j),j}'\\
\vdots\\
H_{Nj}'\\
}},\;0<\abs j<p;
\qq
\bn{\bld C_2}=\brc{\matr{c}{
\wt S_{1,1+p}\wt S_{1+p,1}\\
\vdots\\
\wt S_{N,N+p}\wt S_{N+p,N}\\
}}.}
Also, in (\ref{eq:Jacobi}) we denote

\eq{\fl\bld X(j)=\bfi{X_{i-(Nj)\bmod p,i},\;i=1+mp,\;0\leqslant m\leqslant N-k(j)},\quad 0<j<p,}

\eq{\fl\bld X(-j)=\bfi{X_{i,i-(Nj)\bmod p},\;i=1+mp,\;0\leqslant m\leqslant N-k(j)},\quad 0<j<p,}
and sort out variables in amount equal to the number of elements in \eqs{\bld H'(\pm j)} so that \eqs{\bld A(j)} are square matrices,

\eq{\bfis{\wt S_{i,i+p}}=\bfis{\wt S_{i,i+p},\;1\leqslant i\leqslant N}.}

Due to the fact that every diagonal block of the matrix (\ref{eq:Jacobi}) is square, we can calculate the determinant of (\ref{eq:Jacobi}). The first diagonal block is a square Vandermonde matrix

\eq{\bld V=\brc{\matr{ccc}{
1&\dots&1\\
l_1&\dots&l_N\\
\vdots&&\vdots\\
l_1^{N-1}&\dots&l_N^{N-1}\\
}}}
with the well-known determinant \eqs{\ds\det V=\prod_{1\leqslant i< j\leqslant N}\brc{l_j-l_i}}.

The blocks \eqs{\bld A(j)} are more complicated. Every element \eqs{A_{mn}(j)} of \eqs{\bld A(j)} is a complete homogeneous symmetric polynomial (see \cite{Macdonald}) \eqs{h_{m-1}(n,j)} in \eqs{\bfi{l_{1+(n-1)p+rp},\;0\leqslant r\leqslant k(j)-1}}, i.e.,

\eq{
A_{mn}(j)=h_{m-1}(n,j),
\qq
\sum_{m=0}^{+\infty}h_m(n,j)t^m=\prod_{r=0}^{k(j)-1}\brc{1-t\, l_{1+(n-1)p+rp}}^{-1}.}
The determinants of these matrices are

\eq{
\det\bld A(j)=\prod_{i_1>i_2}\brc{l_{i_1}-l_{i_2}},
}
\eq{i_1=\bfi{1+mp,\quad k(j)\leqslant m\leqslant 2\brc{k(j)-1}},
\qq
i_2=\bfi{1+mp,\quad 0\leqslant m\leqslant k(j)},}
(\ref{eq:symdet}) (see 
\ref{app:det}).
The last two diagonal blocks in (\ref{eq:Jacobi}) are simple:

\eq{\ds b=\pd{H'_{Np}}{\wt S_{pN}}=\pd{H_{Np}}{\wt S_{pN}}\propto\prod_{i\neq p}\wt S_{i,i-p},}
\eq{\bld D=\textrm{diag}\bfi{\wt S_{1+p,1},\wt S_{2+p,2},\dots,\wt S_{pN}}.}
Thus, the determinant of the matrix (\ref{eq:Jacobi}) is the following product:

\eq{b \det\bld V \det\bld D\prod_{j=1}^{p-1}\brc{\det\bld A(j)}^2}
and the Jacobian matrix has the maximum rank.

\end{proof}

Now we are to prove the main statement of this subsection.

\begin{prop}
If $N$ and $p$ are relatively prime, then the systems under consideration are completely integrable in the Liouville sense.
\end{prop}

\begin{proof}

Poisson algebra of the systems can be separated into two disjoint subalgebras. As it was stated earlier in Subsection \ref{sec:lma} the elements of the first subalgebra are the functions of the elements of the second one on the generic symplectic submanifold. The submanifold corresponding to the elements of the second subalgebra has the dimension \eqs{\brc{2p+1}N-1}.
There is the following condition
(see 
\ref{app:dsl}) for the dimension $R$ of the symplectic leaf of this submanifold:

\eq{R\geqslant 2p(N-1).}

Propositions \ref{prop:casimirs} and \ref{prop:independence} give us \eqs{2p+N-1} independent Casimir functions and, consequently, we have the equality

\eq{R=2p(N-1).}

According to the Liouville theorem about integrable systems it is necessary to have \eqs{R/2} functionally independent Hamiltonians in involution for the complete integrability. From Propositions \ref{prop:involution}, \ref{prop:independence}, and (\ref{eq:pyramid}) we have \eqs{p(N-1)=R/2} independent Hamiltonians in involution.

\end{proof}

\section{Conclusion}

We have proposed a procedure giving a limit relation between the elliptic \eqs{SL(N,\mathbb C)} top and the Toda chains. This procedure is similar to the Inozemtsev limit and is a combination of the shift of the spectral parameter, the scalings of coordinates and the trigonometric limit.

Also, in Subsection \ref{sec:somesytem} we have shown, that the generalization (\ref{eq:perscaling}) of the above procedure provides a new class of integrable systems in the case when $N$ and $p>1$ are relatively prime. The open problem is to understand whether the limiting systems are integrable in general case when $p<N/2$. Some statements, such as Propositions \ref{prop:trform}, \ref{prop:involution}, \ref{prop:casimirs}, are still valid in general case, but the whole set of independent Casimir functions and Hamiltonians is not clear.

\section*{Acknowledgements}

We would like to thank A.~M.~Levin and A.~V.~Zotov for many fruitful discussions and A.~Smirnov for useful remarks. We are especially grateful to M.~A.~Olshanetsky for initiating the work and fruitful discussions. The work of both authors was supported by grants RFBR 09-02-00-393, RFBR Consortium E.I.N.S.T.E.IN 09-01-92437-CE and by Federal Agency for Science and Innovations of Russian Federation under contract 14.740.11.0347.

\appendix

\section{Sin-Algebra}
\label{app:veylbasis}

We will use the following notation to simplify formulae:

\eq{\dd\brc{n}=\left\{\matr{ll}{
1,&n\equiv0\bmod N,\\
0,&n\not\equiv0\bmod N,\\
}\right.}

\eq{\bld\rme(z)=\exp\brc{2\pi\rmi z}.}

Elements \eqs{T_{mn}} of the Sin-Algebra basis in \eqs{\mathfrak{sl}\brc{N,\mathbb C}} can be written in the form

\eq{\brc{T_{mn}}_{ij}=\bld\rme\brc{\dfrac{mn}{2N}}\bld\rme\brc{\dfrac{im}N}\dd\brc{j-i-n},
\qq
m\neq0\quad\textrm{or}\quad n\neq0,\qq m,n\in\{0\dots N-1\}.}
For $m,n\in\mathbb Z,$\quad $\brc{m\not\equiv0\bmod N}\;\textrm{or}\;\brc{n\not\equiv0\bmod N},$ the quasi-periodic condition can be introduced

\eq{T_{mn}=\bld\rme\brc{\dfrac{mn-(m\bmod N)(n\bmod N)}{2N}}T_{m\bmod N,n\bmod N},}

\eq{s_{mn}=\bld\rme\brc{\dfrac{(m\bmod N)(n\bmod N)-mn}{2N}}s_{m\bmod N,n\bmod N},}
where \eqs{\bld\rme\brc{\brc{mn-(m\bmod N)(n\bmod N)}/\brc{2N}}=\pm1.}

The commutator relations in this basis are

\eqlb{eq:weilcomm}{\bsq{T_{mn},T_{kl}}=2\rmi\sin\bsq{\dfrac\pi N\brc{kn-ml}}T_{m+k,n+l}.}

One can establish the following relations between the coordinates in the standard \eqs{\bfi{S_{ij}}} and Sin-Algebra \eqs{\bfi{s_{mn}}} bases

\eqlb{eq:weiltr}{S_{ij}=\sum_{m,n}s_{mn}\brc{T_{mn}}_{ij},\qq s_{mn}=\dfrac1N\sum_{i,j}S_{ij}\brc{T_{-m,n}}_{ij}.}

\section{Degenerations of elliptic functions}
\label{app:ellipticfunctions}

Definitions and properties of elliptic functions are borrowed mainly from \cite{Veyl1} and \cite{Mumford}. The main object is the theta function with characteristics defined via

\eq{\theta\bsq{\at a b}\brc{z,\tau}=\sum_{j\in \mathbb Z}q^{\frac12 (j+a)^2}\bld\rme\brc{(j+a)(z+b)},}
where \eqs{q=\bld\rme\brc{\tau}\equiv\exp\brc{2\pi\rmi\tau}}.

We will also need the Eisenstein functions


\eqlb{eq:eisdef}{\varepsilon_k(z)=\lim_{M\rightarrow+\infty}\sum_{n=-M}^M(z+n)^{-k},
\qq
E_k(z)=\lim_{M\rightarrow+\infty}\sum_{n=-M}^M\varepsilon_k(z+n\tau).}

To determine the limits of Lax matrices we will use the series expansions of the following functions
\eqlb{eq:thetadef}{\vartheta(z)=\theta\bsq{\at{1/2}{1/2}}\brc{z,\tau}=\sum_{j\in\mathbb Z}q^{\frac12\brc{j+\frac12}^2}\bld\rme\brc{\brc{j+\dfrac12}\brc{z+\dfrac12}},}

\eqlb{eq:phidef}{\phi(u,z)=\dfrac{\vartheta(u+z)\vartheta'(0)}{\vartheta(u)\vartheta(z)},}

\eqlb{eq:laxcoeff}{\ds\varphi\bsq{\at mn}(z)=\bld\rme\brc{-\dfrac{nz}N}\phi(-\dfrac{m+n\tau}N,z),\qq
\ds f\bsq{\at mn}(z)=\bld\rme\brc{-\dfrac{nz}N}\partial_u\phi(u,z)|_{u-\frac{m+n\tau}N.}}

The functions satisfy the following well-known identities:

\eq{\phi(u,z)\phi(-u,z)=E_2(z)-E_2(u),}

\eqlb{eq:pdphi}{\partial_u\phi(u,z)=\phi(u,z)(E_1(u+z)-E_1(u)),}
parity

\eq{E_k(-z)=(-1)^k E_k(z),}
\eq{\vartheta(-z)=-\vartheta(z),}
\eq{\phi(u,z)=\phi(z,u)=-\phi(-u,-z),}
and quasi-periodicity

\eqlb{eq:quasiperiodicity}{\matr{ll}{
E_1(z+1)=E_1(z),&E_1(z+\tau)=E_1(z)-2\pi \rmi,\\
\\
E_2(z+1)=E_2(z),&E_2(z+\tau)=E_2(z),\\
\\
\vartheta(z+1)=-\vartheta(z),&\vartheta(z+\tau)=-q^{-\frac12}\bld\rme(-z)\vartheta(z),\\
\\
\phi(u+1,z)=\phi(u,z),&\phi(u+\tau,z)=\bld\rme(-z)\phi(u,z).}}

We will examine degenerations of elliptic functions (\ref{eq:laxcoeff}) in the following limit:

\eq{z=\wt z+\dfrac\tau2,\qq Im(\tau)\rightarrow+\infty.}

Using definition (\ref{eq:phidef}) one can reduce the expansion of \eqs{\varphi\bsq{\at mn}(z)} to the expansion of theta functions. Considering the main non-vanishing terms we have

\eq{\fl\vartheta\brc{-\dfrac mN-\dfrac nN\tau}=\left\{\matr{ll}{
2q^{\frac18}\sin\brc{\pi\dfrac mN}+\bld{\mathrm o}\brc{q^{\frac18}},&n=0,\\
\\
iq^{\frac18-\frac n{2N}}\bld\rme\brc{-\dfrac m{2N}}+\bld{\mathrm o}\brc{q^{\frac18-\frac n{2N}}},&0<n<N,\\
}\right.}

\eq{\fl\vartheta\brc{\wt z+\dfrac\tau2-\dfrac mN-\dfrac nN\tau}=
\left\{\matr{ll}{
-\rmi q^{\frac n{2N}-\frac18}\bld\rme\brc{\dfrac12\brc{\dfrac mN-\wt z}}+\bld{\mathrm o}\brc{q^{\frac n{2N}-\frac18}},&0\leqslant n<\dfrac N2,\\
\\
-2q^{\frac18}\sin\brc{\pi\brc{\wt z-\dfrac mN}}+\bld{\mathrm o}\brc{q^{\frac18}},&n=\dfrac N2,\\
\\
\rmi q^{\frac38-\frac n{2N}}\bld\rme\brc{\dfrac12\brc{\wt z-\dfrac mN}}+\bld{\mathrm o}\brc{q^{\frac38-\frac n{2N}}},&\dfrac N2<n<N,\\
}\right.}
wich gives

\eqlb{eq:phidevel}{\fl\phi\brs{-\dfrac{m+n\tau}N;\wt z+\dfrac\tau2}=
\left\{\matr{ll}{
-\pi\bld\rme\brc{\dfrac m{2N}}\sin^{-1}\brc{\pi\dfrac mN}+\bld{\mathrm o}\brc{1},&n=0,\\
\\
2\pi\rmi q^{\frac nN}\bld\rme\brc{\dfrac mN}+\bld{\mathrm o}\brc{q^{\frac nN}},&0<n<\dfrac N2,\\
\\
4\pi q^{\frac12}\sin\brc{\pi\brc{\wt z-\dfrac mN}}\bld\rme\brc{\dfrac m{2N}+\dfrac12\wt z}+\bld{\mathrm o}\brc{q^{\frac12}},&n=\dfrac N2,\\
\\
-2\pi\rmi q^{\frac12}\bld\rme\brc{\wt z}+\bld{\mathrm o}\brc{q^{\frac12}},&\dfrac N2<n<N,\\
}\right.}

\eqlb{eq:varphiexp}{\fl
\varphi\bsq{\at mn}\brc{\wt z+\dfrac\tau2}=
\left\{\matr{ll}{
-\pi\bld\rme\brc{\dfrac m{2N}}\sin^{-1}\brc{\pi\dfrac mN}+\bld{\mathrm o}\brc{1},&n=0,\\
\\
2\pi\rmi q^{\frac n{2N}}\bld\rme\brc{\dfrac mN-\dfrac{n\wt z}N}+\bld{\mathrm o}\brc{q^{\frac n{2N}}},&0<n<\dfrac N2,\\
\\
4\pi q^{\frac14}\bld\rme\brc{\dfrac m{2N}}\sin\brc{\pi\brc{\wt z-\dfrac mN}}+\bld{\mathrm o}\brc{q^{\frac14}},&n=\dfrac N2,\\
\\
-2\pi\rmi q^{\frac{N-n}{2N}}\bld\rme\brc{\dfrac{N-n}N\wt z}+\bld{\mathrm o}\brc{q^{\frac{N-n}{2N}}},&\dfrac N2<n<N.\\
}\right.}

To evaluate the limit of \eqs{f\bsq{\at m n}} we expand \eqs{E_1\brc{\wt x-\sigma\tau}} as a series in $q$. From the definition (\ref{eq:eisdef}) we get

\eqa{&& E_1\brc{\wt x-\sigma\tau}=\lim_{M\rightarrow+\infty}\sum_{n=-M}^{M}\varepsilon_1\brc{\wt x+\brc{n-\sigma}\tau}=
\varepsilon_1\brc{\wt x-\sigma\tau}+
\cr
&&+\lim_{M\rightarrow+\infty}\sum_{n=1}^{M}
\brc{\varepsilon_1\brc{\wt x+\brc{n-\sigma}\tau}+\varepsilon_1\brc{\wt x-\brc{n+\sigma}\tau}}.}

Using the explicit formula for \eqs{\varepsilon_1(x)} from \cite{Veyl1}

\eq{\fl\varepsilon_1(x)=\pi\cot\brc{\pi x}
=\pi\rmi\dfrac{\bld\rme(x)+1}{\bld\rme(x)-1}=\pi\rmi\left\{\matr{rl}{
-1-2\bld\rme(x)+\bld{\mathrm o}\brc{\bld\rme(x)},&Im(x)\rightarrow +\infty,\\
1+2\bld\rme(x)+\bld{\mathrm o}\brc{\bld\rme(x)},&Im(x)\rightarrow +\infty,\\
}\right.}
one can write down the leading term of \eqs{E_1(\wt x -\sigma\tau)} as follows:

\eq{\fl E_1(\wt x-\sigma\tau)=\left\{\matr{ll}{
\pi\cot(\pi\wt x)+\bld{\mathrm o}\brc{1},&
\sigma=0,\\
\pi\rmi+2\pi\rmi q^\sigma\bld\rme(-\wt x)+\bld{\mathrm o}\brc{q^\sigma},&
0<\sigma<\dfrac12,\\
\pi\rmi+2\pi\rmi q^{\frac12}\brc{\bld\rme(-\wt x)-\bld\rme(\wt x)}+\bld{\mathrm o}\brc{q^{\frac12}},&
\sigma=\dfrac12,\\
\pi\rmi-2\pi\rmi q^{1-\sigma}\bld\rme\brc{\wt x}+\bld{\mathrm o}\brc{q^{1-\sigma}},&
\dfrac12<\sigma<1,
}\right.}
and using (\ref{eq:quasiperiodicity}) generalize it to \eqs{\sigma\in\mathbb R}:

\eq{\fl E_1(\wt x-\sigma\tau)=\left\{\matr{ll}{
2\pi\rmi\flr{\sigma}+\pi\cot(\pi\wt x)+\bld{\mathrm o}\brc{1},&
\{\sigma\}=0,\\
2\pi\rmi\flr{\sigma}+\pi\rmi+2\pi\rmi q^{\{\sigma\}}\bld\rme(-\wt x)+\bld{\mathrm o}\brc{q^{\{\sigma\}}},&
0<\{\sigma\}<\dfrac12,\\
2\pi\rmi\flr{\sigma}+\pi\rmi+2\pi\rmi q^{\frac12}\brc{\bld\rme(-\wt x)-\bld\rme(\wt x)}+\bld{\mathrm o}\brc{q^{\frac12}},&
\{\sigma\}=\dfrac12,\\
2\pi\rmi\flr{\sigma}+\pi\rmi-2\pi\rmi q^{1-\{\sigma\}}\bld\rme\brc{\wt x}+\bld{\mathrm o}\brc{q^{1-\{\sigma\}}},&
\dfrac12<\{\sigma\}<1,
}\right.}
where \eqs{\bfi{\sigma}} is the fractional part of $\sigma$.

To expand \eqs{\partial_u\phi(u,z)|_{u=\wt u-\sigma\tau}} in the limit \eqs{Im(\tau)\rightarrow+\infty} it is convenient to use formula (\ref{eq:pdphi}). Assuming \eqs{z=\wt z+\tau/2} we have to consider the following cases depending on the value of $\sigma$:
\begin{enumerate}
\item{\eqs{\sigma=0}}
    \eq{
    \phi\brs{\wt u,\wt z+\dfrac\tau2}=\dfrac{\pi\bld
    e\brs{-\dfrac12\wt u}}{\sin\pi\wt u}+\bld{\mathrm o}\brc{1},
    }
    \eq{E_1\brs{\wt u}=\pi\cot\pi\wt u+\bld{\mathrm o}\brc{1},
    }
    \eq{E_1\brs{\wt u+\wt z+\dfrac\tau2}=-\pi\rmi+\bld{\mathrm o}\brc{1},}
    \eq{\Downarrow}
    \eq{\partial_u\phi(u,z)|_{u=\wt u-\sigma\tau}=-\pi^2\sin^{-2}\pi\wt u+\bld{\mathrm o}\brc{1}.}
\item{\eqs{0<\sigma<\frac12}}
    \eq{
    \phi\brs{\wt u-\sigma\tau,\wt z+\dfrac\tau2}=2\pi\rmi q\bld\rme(-\wt u)+\bld{\mathrm o}\brc{q},
    }
    \eq{E_1(\wt u-\sigma\tau)=\pi\rmi+2\pi\rmi q\bld\rme(-\wt u)+\bld{\mathrm o}\brc{1},
    }
    \eq{E_1\brs{\wt u+\wt z+\brs{\dfrac12-\sigma}\tau}=
    -\pi\rmi-2\pi\rmi q^{\frac12-\sigma}\bld\rme\brs{\wt u+\wt z}+\bld{\mathrm o}\brc{q^{\frac12-\sigma}},}
    \eq{\Downarrow}
    \eq{\partial_u\phi(u,z)|_{u=\wt u-\sigma\tau}=4\pi^2q\bld\rme\brs{-\wt u}+\bld{\mathrm o}\brc{q}.}
\item{\eqs{\sigma=\frac12}}
    \eq{
    \phi\brs{u,z}=4\pi q^{\frac12}\sin\brs{\pi\brs{\wt u+\wt z}}\bld\rme\brs{-\dfrac12\wt u+\dfrac12\wt z}+\bld{\mathrm o}\brc{q^{\frac12}},
    }
    \eq{E_1\brs{\wt u+\wt z}=\pi\cot\brs{\pi\brs{\wt u+\wt z}}+\bld{\mathrm o}\brc{1},
    }
    \eq{E_1\brs{\wt u-\dfrac12\tau}=\pi\rmi+\bld{\mathrm o}\brc{1},}
    \eq{\Downarrow}
    \eq{\partial_u\phi(u,z)|_{u=\wt u-\sigma\tau}=4\pi^2q^{\frac12}\bld\rme\brc{-\wt u}+\bld{\mathrm o}\brc{q^{\frac12}}.}
\item{\eqs{\frac12<\sigma<1}}
    \eq{
    \phi\brs{\wt u-\sigma\tau,\wt z+\dfrac\tau2}=
    -2\pi\rmi q^{\frac12}\bld\rme\brs{\wt z}+\bld{\mathrm o}\brc{q^{\frac12}},
    }
    \eq{E_1\brs{\wt u-\sigma\tau}=\pi\rmi-2\pi\rmi q^{1-\sigma}\bld\rme\brs{\wt u},
    }
    \eq{E_1(\wt u+\wt z+\brs{\dfrac12-\sigma}\tau)=\pi\rmi+2\pi\rmi q^{\sigma-\frac12}\bld\rme\brs{-\wt u-\wt z}+\bld{\mathrm o}\brc{q^{\sigma-\frac12}},}
    \eq{\Downarrow}
    \eqa{&&\fl\partial_u\phi(u,z)|_{u=\wt u-\sigma\tau}=4\pi^2 q^{\frac12}\bsq{\bld\rme\brs{-\wt u-\wt z+\brs{\sigma-\dfrac12}}-\bld\rme\brs{\wt u+\brs{1-\sigma}\tau}}\bld\rme\brs{\wt z}=\cr
    &&=\left\{\matr{ll}{
    4\pi^2 q\bld\rme\brc{-\wt u}+\bld{\mathrm o}\brc{q},&\dfrac12<\sigma<\dfrac34,\\
    4\pi^2 q^{\frac34}\brc{\bld\rme\brc{-\wt u}-\bld\rme\brc{\wt u+\wt z}},&\sigma=\dfrac34,\\
    -4\pi^2q^{\frac32-\sigma}\bld\rme\brc{\wt u+\wt z},&\dfrac34<\sigma<1.\\
    }\right.}
\end{enumerate}

Summarizing all the special cases we obtain

\eq{\fl\partial_u\phi(u,z)|_{u=\wt u-\sigma\tau}=\left\{\matr{ll}{
-\pi^2\sin^{-2}\pi\wt u+\bld{\mathrm o}\brc{1},&\sigma=0,\\
4\pi^2q\bld\rme\brc{-\wt u}+\bld{\mathrm o}\brc{q},&0<\sigma<\dfrac34,\\
4\pi^2q^{\frac34}\bsq{\bld\rme\brs{-\wt u}-\bld\rme\brs{\wt u+\wt z}}+\bld{\mathrm o}\brc{q^{\frac34}},&\sigma=\dfrac34,\\
-4\pi^2q^{\frac32-\sigma}\bld\rme\brs{\wt u+\wt z}+\bld{\mathrm o}\brc{q^{\frac32-\sigma}},&\dfrac34<\sigma<1.
}\right.}

Finally, from (\ref{eq:laxcoeff}) we get

\eqlb{eq:fexp}{\fl f\bsq{\at mn}\brs{\wt z+\dfrac\tau2}
=\left\{\matr{ll}{
-\pi^2\sin^{-2}\brc{\pi\dfrac mN}+\bld{\mathrm o}(1),&n=0,\\
\\
4\pi^2\bld\rme\brc{\dfrac mN}\bld\rme\brc{-\dfrac{n\wt z}N}q^{\frac n{2N}}+\bld{\mathrm o}\brc{q^{\frac n{2N}}},&0<n<\dfrac{3N}4,\\
\\
4\pi^2\bsq{\bld\rme\brc{\dfrac mN}-\bld\rme\brc{-\dfrac nN+\wt z}}\bld\rme\brc{-\dfrac34\wt z}q^{\frac38}+\bld{\mathrm o}\brc{q^{\frac38}},&n=\dfrac{3N}4,\\
\\
-4\pi^2\bld\rme\brc{-\dfrac mN+\wt z}\bld\rme\brc{-\dfrac{n\wt z}N}q^{\frac32\brc{1-\frac nN}}+\bld{\mathrm o}\brc{q^{\frac32\brc{1-\frac nN}}},&\dfrac{3N}4<n<N.\\
}\right.}

\section{Scaling inequality}
\label{app:inequality}

We are going to prove the inequality
\eqs{\alpha_i+\alpha_j-\alpha_{i+j}\geqslant0,\quad\forall i,j\in\mathbb Z},
for the function

\eq{\alpha_i=\left\{\matr{ll}{
\dfrac k{2N},&0\leqslant k<p\leqslant\dfrac N2,\\
\\
\dfrac p{2N},&p\leqslant k\leqslant N-p,\\
\\
\dfrac12-\dfrac k{2N},&N-p<k<N,\\
}\right.}
where \eqs{k\equiv i\bmod N}.

Taking into account that \eqs{\alpha_{i+N}=\alpha_i}, it is enough to consider the case \eqs{i,j\in\bsq{-N/2;N/2}}, where we have

\eq{\alpha_i=\left\{\matr{ll}{
\dfrac{\abs i}{2N},&\abs i\leqslant p,\\
\\
\dfrac p{2N},&p\leqslant\abs i\leqslant\dfrac N2.\\
}\right.}

It is easy to check that there are the following two cases:
\begin{enumerate}
\item{If \eqs{\brc{\abs i> p}\vee\brc{\abs j> p},} then}

\eq{\alpha_{i+j}\leqslant\dfrac p{2N}\leqslant\alpha_i+\alpha_j.}

\item{If \eqs{\brc{\abs i\leqslant p}\wedge\brc{\abs j\leqslant p},} then}

\eq{\alpha_{i+j}\leqslant\dfrac{\abs{i+j}}{2N}\leqslant\dfrac{\abs i}{2N}+\dfrac{\abs j}{2N}=\alpha_i+\alpha_j.}

\end{enumerate}

\section{Dimension of the symplectic leaf}
\label{app:dsl}

The phase space of the system from Subsection \ref{sec:somesytem} is equipped with the following Poisson structure:

\eqa{
&&\{\wt S_{ii},\wt S_{jk}\}=N(\wt S_{ji}\delta_{ik}-\wt S_{ik}\delta_{ij}),
\cr\cr
&&\fl\{\wt S_{ij},\wt S_{kl}\}=N(\wt S_{kj}\delta_{il}-\wt S_{il}\delta_{kj}),
\cr
&&(0<(j-i)\bmod N\leqslant p)\wedge(0<(l-k)\bmod N\leqslant p)
\wedge(0<(j+l-i-k)\bmod N\leqslant p),
\cr\cr
&&\{\wt S_{ij},\wt S_{kl}\}=N(\wt S_{kj}\delta_{il}-\wt S_{il}\delta_{kj}),
\cr
&&(N-p\leqslant(j-i)\bmod N< N)\wedge(N-p\leqslant(l-k)\bmod N<N)\wedge
\cr
&&\wedge(N-p\leqslant(j+l-i-k)\bmod N<N).
}

We are interested in the specific subalgebra
\eq{\fl\mathcal S_2=\bfi{\wt S_{ij},\;(0\leqslant\brc{j-i}\bmod N\leqslant p)\;\mathrm{or}\;(N-p\leqslant \brc{j-i}\bmod N<N)}}
and the restriction of the Poisson tensor \eqs{\pi^{(ij)(kl)}\brs{\mathcal S_2}} on it

\eq{\bfi{F\brs{\mathcal S_2},G\brs{\mathcal S_2}}=\pi^{(ij)(kl)}\brs{\mathcal S_2}\;\partial_{(ij)}F\;\partial_{(kl)}G,\qq \partial_{(ij)}=\pd{}{\wt S_{ij}}.}

The dimension of the symplectic leaf of the Poisson submanifold \eqs{\mathcal S_2} is the rank $R$ of \eqs{\pi^{(ij)(kl)}\brs{\mathcal S_2}}. Our aim is to minorize the rank $R$. Note that \eqs{\pi^{(ij)(kl)}\brs{\mathcal S_2}} can be represented as a square \eqs{(2p+1)N\times(2p+1)N}-matrix. To write down this matrix in a block triangular form we use the following ordering:

\eq{\bld Y=\bfi{\bld Y(k),\; k=\pm p,\dots,\pm1,0},
\qq \bld Y(k)=\bfi{\wt S_{i,i+k},\; i=1,\dots,N}.}

Thus, we have

\arraycolsep=0pt
\eqlb{eq:pitens}{\fl\mfs{\scriptsize}{
\bordermatrix{
&\bm{Y}(p)\quad\bm{Y}(-p)&\bm{Y}(p-1)\;\bm{Y}(1-p)&\dots&\bm{Y}(1)\;\bm{Y}(-1)&\bm{Y}(0)\cr
\cr
\matr{l}{\bm{Y}(p)\\ \bm{Y}(-p)}
&\bm{0}&\bm{0}&\bm{0}&\bm{0}&\pmat{c}{\bm{P_+}(0)\\ \bm{P_-}(0)}\cr
\cr
\matr{l}{\bm{Y}(p-1)\\ \bm{Y}(1-p)}&\bm{0}&\bm{0}&\bm{0}&\pmat{cc}{\bm{P_+}(1)&\bm{0}\\ \bm{0}&\bm{P_-}(1)}&\vdots\cr
\cr
\matr{l}{\vdots}&\bm{0}&\bm{0}&\adots&\adots&\vdots\cr
\cr
\matr{l}{\bm{Y}(1)\\ \bm{Y}(-1)}&\bm{0}&\pmat{cc}{\bm{P_+}(p-1)&\bm{0}\\ \bm{0}&\bm{P_-}(p-1)}&\adots&\adots&\vdots\cr
\cr
\matr{l}{\bm{Y}(0)}&\pmat{cc}{\bm{P_+}(p)\;&\bm{P_-}(p)}&\dots&\dots&\dots&\bm{0}
}},}
where

\eqa{\fl(P_+(k))_{ij}=\bfi{\wt S_{i,i+p-k},\wt S_{j,j+k}}
=N\brc{\wt S_{j,j+p}\delta_{i,j+k}-\wt S_{i,i+p}\delta_{j,i+p-k}},\qq
k\in\bfi{0,\dots,p},
\cr
\fl(P_-(k))_{ij}=\bfi{\wt S_{i,i+k-p},\wt S_{j,j-k}}
=N\brc{\wt S_{j,j-p}\delta_{i,j-k}-\wt S_{i,i-p}\delta_{j,i+k-p}},\qq
k\in\bfi{0,\dots,p}.}

It can be easily seen that the rank of each matrix \eqs{\bld P_{\pm}(k)} is equal to \eqs{N-1}. Due to the block triangular form of the matrix (\ref{eq:pitens}) we obtain the required condition
\eqlb{eq:symprank}{R\geqslant 2p(N-1).}

\section{Addition to Proposition \ref{prop:independence}, \eqs{\det\bld A(j)}}
\label{app:det}

Here we are to consider square matrices \eqs{\bld A(j)} with the following structure: every element \eqs{A_{mn}(j)} of \eqs{\bld A(j)} is a complete homogeneous symmetric polynomial (see \cite{Macdonald}) \eqs{h_{m-1}(n,j)} in variables 
\eqs{\bfi{l_{1+(n-1)p+rp},\;0\leqslant r\leqslant k(j)-1}}, i.e.,

\eqlb{eq:expel}{A_{mn}(j)=h_{m-1}(n,j),\qq
\sum_{m=0}^{+\infty}h_m(n,j)t^m=\prod_{r=0}^{k(j)-1}\brc{1-t\, l_{1+(n-1)p+rp}}^{-1},}
where we use the notation \eqs{k(j)=\cl{jN/p}} from Proposition \ref{prop:independence}.

The determinants of the matrices under consideration are homogeneous polynomials of order \newline \eqs{\brc{k(j)(k(j)-1)}/2}. In order to compute these determinants we will find the appropriate number of roots and calculate their values at one particular point.

As one can see from (\ref{eq:expel}) columns in matrix \eqs{\bld A(j)} depend on the following set of variables:

\eqlb{eq:matvarset}{
\arraycolsep=0.5em
\fl\matr{cc}{
\textrm{Column Number}&\textrm{Set of Variables}\\
\matr{c}{1\\2\\\vdots\\k(j)\\}&
\matr{ccccccc}{
l_1&l_{1+p}&\dots&l_{1+(k(j)-1)p}&\\
&l_{1+p}&\dots&l_{1+(k(j)-1)p}&l_{1+k(j)p}\\
&&\ddots&\dots&\dots&\ddots\\
&&&l_{1+(k(j)-1)p}&l_{1+k(j)p}&\dots&l_{1+2(k(j)-1)p}\\
}\\}}

It follows from (\ref{eq:matvarset}) that two adjacent columns will coincide if we put the first variable, e.g., $l_1$, in one column equal to the last variable, e.g., \eqs{l_{1+k(j)p}}, in the subsequent column. Thus, we have \eqs{\cl{jN/p}-1} roots. It turns out that the result can be generalized to any $c$ adjacent columns, where \eqs{2\leqslant c\leqslant k(j)=\cl{jN/p}}. In other words, if the first variable in the first column is equal to the last variable in the last column, then the determinant is equal to zero. Indeed, without loss of generality we can consider the first $c$ columns. By assumption,

\eq{l_{1+(k(j)+c-2)p}=l_1.}

We will prove that in this case the rows of \eqs{\bld A(j)} are linearly dependent. If we multiply the $m$'th row of matrix  by the elementary symmetric polynomial \eqs{(-1)^{k(j)-m} e_{k(j)-m}(j)} in variables \newline
\eqs{\bfi{l_{1+rp},\;0\leqslant r\leqslant 2k(j)-2,\;r\neq k(j)+c-2}}

\eq{\sum_{m\geqslant 0}(-1)^m e_m(j) t^m=\prod_{r=0,r\neq k(j)+c-2}^{2k(j)-2}\brc{1-t\,l_{1+rp}},}
then after summation we get the zero row due to the fact that the elements of the row are coefficients in front of \eqs{t^{k(j)-1}} in the following polynomial with the maximum degree \eqs{k(j)-2}:

\eqa{
&&\sum_{m\geqslant 0}(-1)^m e_m(j) t^m\sum_{m_1\geqslant0}h_{m_1}(n,j) t^{m_1}=
\cr
&&=\left\{\matr{ll}{
\brc{\prod_{r=0}^{n-2}\brc{1-t\, l_{1+rp}}}\brc{\prod_{r=k(j)+n-1,r\neq k(j)+c-2}^{2k(j)-2}\brc{1-t\, l_{1+rp}}},\qq&n<c,\\
\\
\brc{\prod_{r=1}^{n-2}\brc{1-t\, l_{1+rp}}}\brc{\prod_{r=k(j)+n-1}^{2k(j)-2}\brc{1-t\, l_{1+rp}}},&n\geqslant c.\\
}\right.}

Thus, we have \eqs{\brc{k(j)(k(j)-1)}/2} roots. Now we calculate the determinant at the following point
\eq{l_{i_1}=1,\qq l_{i_2}=0,\qq A_{mn}=\brc{\at{m+n-3}{m-1}},
}
\eq{i_1=\bfi{1+mp,\quad k(j)\leqslant m\leqslant 2\brc{k(j)-1}},
\qq
i_2=\bfi{1+mp,\quad 0\leqslant m\leqslant k(j)}.
}

Subtracting columns according to the well-known formula, i.e.,

\eq{\brc{\at{N+1}{m+1}}=\brc{\at Nm}+\brc{\at N{m+1}}}
\eq{\Downarrow}
\eq{A_{m+1,n}-A_{m+1,n-1}=A_{mn},}
one can obtain a lower unit triangular matrix, so at the point selected above we have

\eq{\det\bld A = 1.}

Finally, we get

\eqal{eq:symdet}{
&&\det\bld A(j)=\prod_{i_1>i_2}\brc{l_{i_1}-l_{i_2}},
\cr\cr
&&i_1=\bfi{1+mp,\quad k(j)\leqslant m\leqslant 2\brc{k(j)-1}},
\qq
i_2=\bfi{1+mp,\quad 0\leqslant m\leqslant k(j)}.}

\bibliographystyle{unsrt}
\bibliography{references}

\begin{thebibliography}{10}

\bibitem{Krichever1}
I.~M. Krichever.
\newblock Integration of nonlinear equations by the methods of algebraic
  geometry.
\newblock {\em Functional Analysis and Its Applications}, 11(1):12--26, 1977.

\bibitem{Krichever2}
I.~M. Krichever.
\newblock Algebraic curves and commuting matricial differential operators.
\newblock {\em Functional Analysis and Its Applications}, 10(2):144--146, 1976.

\bibitem{Dubrovin}
B.~A. Dubrovin.
\newblock Completely integrable {H}amiltonian systems associated with matrix
  operators and {A}belian varieties.
\newblock {\em Functional Analysis and Its Applications}, 11(4):265--277, 1977.

\bibitem{Inozemtsev}
V.~I. Inozemtsev.
\newblock The finite toda lattices.
\newblock {\em Communications in Mathematical Physics}, 121(4):629--638, 1989.

\bibitem{ZotovChernyakov}
Yu.~B. Chernyakov and A.~V. Zotov.
\newblock Integrable many-body systems via the {I}nosemtsev limit.
\newblock {\em Theoretical and Mathematical Physics}, 129(2):1526--1542, 2001.

\bibitem{LOZ}
A.~M. Levin, M.~A. Olshanetsky, and A.~Zotov.
\newblock Hitchin systems - symplectic {H}ecke correspondence and
  two-dimensional version.
\newblock {\em Communications in Mathematical Physics}, 236(1):93--133, 2003.

\bibitem{Smirnov3}
A.~Smirnov.
\newblock Integrable $\mathfrak{sl}({N},\mathbb{C})$ tops as {C}alogero-{M}oser
  systems.
\newblock {\em Theoretical and Mathematical Physics}, 158(3):300--312, 2009.

\bibitem{GGM1}
A.~Gorsky, S.~Gukov, and A.~Mironov.
\newblock Multiscale \eqs{N = 2} {SUSY} field theories, integrable systems and
  their stringy /brane origin --- {I}.
\newblock {\em Nuclear Physics B}, 517(1-3):409--461, 1998.

\bibitem{GGM2}
A.~Gorsky, S.~Gukov, and A.~Mironov.
\newblock {SUSY} field theories, integrable systems and their stringy /brane
  origin --- {II}.
\newblock {\em Nuclear Physics B}, 518(3):689--713, 1998.

\bibitem{Arnold}
V.~I. Arnold.
\newblock {\em Mathematical Methods of Classical Mechanics}.
\newblock Springer-Verlag, New York and Berlin, 1978.

\bibitem{Veyl1}
Andre Weil.
\newblock {\em Elliptic Functions according to {E}isenstein and {K}ronecker}.
\newblock Springer-Verlag, 1976.

\bibitem{Reynman1}
A.~G. Reiman and M.~A. Semenov-Tyan-Shanskii.
\newblock Lie algebras and {L}ax equations with spectral parameter on an
  elliptic curve.
\newblock {\em Journal of Mathematical Sciences}, 46(1):1631--1640, 1989.

\bibitem{Belavin1}
A.~A. Belavin.
\newblock Dynamical symmetry of integrable quantum systems.
\newblock {\em Nuclear Physics B}, 180(2):189--200, 1981.

\bibitem{BelavinDrinfeld}
A.~A. Belavin and Drinfel'd.
\newblock Solutions of the classical {Y}ang - {B}axter equation for simple lie
  algebras.
\newblock {\em Functional Analysis and Its Applications}, 16(3):159--180, 1982.

\bibitem{KulishSklyanin}
P.~P. Kulish and E.~K. Sklyanin.
\newblock Solutions of the {Y}ang-{B}axter equation.
\newblock {\em Journal of Soviet Mathematics}, 19(5):1596--1620, 1982.

\bibitem{Calogero}
F.~Calogero.
\newblock Solution of the one-dimensional $n$-body problems with quadratic
  and/or inversely quadratic pair potentials.
\newblock {\em Journal of Mathematical Physics}, 12(3):419--436, 1971.

\bibitem{Moser}
J.~Moser.
\newblock Three integrable {H}amiltonian systems connected with isospectral
  deformations.
\newblock {\em Advances in Mathematics}, 16:197--220, 1975.

\bibitem{Manakov1}
S.~V. Manakov.
\newblock {\em Journal of Experimental and Theoretical Physics}, 40:269, 1974.

\bibitem{Flashka1}
H.~Flashka.
\newblock {\em Physical Review}, B9:1924, 1974.

\bibitem{Flashka2}
H.~Flashka.
\newblock {\em Progress of Theoretical Physics}, 51:703, 1974.

\bibitem{Kheshin}
B.~Khesin, A.~Levin, and M.~Olshanetsky.
\newblock Bihamiltonian structures and quadratic algebras in hydrodynamics and
  on non-commutative torus.
\newblock {\em Communications in Mathematical Physics}, 250(3):581--612, 2004.

\bibitem{Macdonald}
I.~G. Macdonald.
\newblock {\em Symmetric Functions and {H}all Polynomials}.
\newblock Oxford University Press, Oxford, 1979.

\bibitem{Mumford}
David Mumford.
\newblock {\em Tata Lectures on Theta {I},{II}}.
\newblock Birkh\"{a}user, Boston, 1983,~1984.

\end{thebibliography}

\end{document}